\documentclass[runningheads,a4paper]{llncs}

\usepackage{times}
\usepackage{amsmath}
\usepackage{stmaryrd}
\usepackage{amssymb}
\usepackage{graphicx}
\usepackage[T1]{fontenc}

\usepackage[ruled]{algorithm2e}

\usepackage{url}
\urldef{\mailsa}\path|{Mingsheng.Ying}@uts.edu.au|
\newcommand{\keywords}[1]{\par\addvspace\baselineskip
\noindent\keywordname\enspace\ignorespaces#1}

\def\>{\ensuremath{\rangle}}
\def\<{\ensuremath{\langle}}

\newtheorem{thm}{Theorem}
\newtheorem{cor}{Corollary}
\newtheorem{lem}{Lemma}
\newtheorem{defn}{Definition}
\newtheorem{prop}{Proposition}
\newtheorem{exam}{Example}

\newcommand {\spa } {{\rm span}}

\begin{document}

\mainmatter

\title{Quantum Recursion and Second Quantisation}
\author{Mingsheng Ying}
\institute{University of Technology, Sydney, Australia\\ and\\ Tsinghua University, China\\
\email{Mingsheng.Ying@uts.edu.au; yingmsh@tsinghua.edu.cn}}
\titlerunning{Quantum Recursion and Second Quantisation}
\authorrunning{Ying}

\maketitle

\begin{abstract}

This paper introduces a new notion of quantum recursion of which the control flow of the computation is quantum rather than classical as in the notions of recursion considered in the previous studies of quantum programming. A typical example is recursive quantum walks, which are obtained by slightly modifying the construction of the ordinary quantum walks. The operational and denotational semantics of quantum recursions are defined by employing the second quantisation method, and they are proved to be equivalent.

\end{abstract}

\section{Introduction}

Recursion is one of the central ideas of computer science. Most programming languages support recursion or at least a special form of recursion such as while-loop. Recursion has also been considered since the very beginning of the studies of quantum programming; for example, Selinger \cite{Se04} introduced the notion of recursive procedure in his functional quantum programming language QPL and defined the denotational semantics of recursive procedures in terms of complete partial orders of super-operators. Termination of quantum while-loops were analysed by Ying and Feng \cite{YF10} in the case of finite-dimensional state spaces. A quantum generalisation of Etessami and Yannakakis's recursive Markov chains was proposed by Feng et. al.~\cite{FYY13}.
But the control flows of all of the quantum recursions studied in the previous literatures are classical because branchings in them are determined by the outcomes of certain quantum measurements, which are classical information. So, they can be appropriately called \textit{classical recursion of quantum programs}.

Quantum control flow was first introduced by Altenkirch and Grattage \cite{AG05} by defining a quantum case statement in their quantum programming language QML that implements a unitary transformation by decomposing it into two orthogonal branches along an orthonormal basis of a chosen qubit. Motivated by the construction of quantum walks \cite{AA01}, \cite{AB01},  a different approach to quantum control flow was proposed by the author in \cite{YYF12}, \cite{YYF13} where a kind of quantum case statement was defined by employing an external quantum \textquotedblleft coin\textquotedblright. Furthermore, the notion of quantum choice was defined as the sequential composition of a \textquotedblleft coin tossing\textquotedblright\ program and a quantum case statement. The quantum control flow of programs is clearly manifested in these quantum case statement and quantum choice.

This paper introduces a new notion of quantum recursion using quantum case statements and quantum choices. 
In contrast to the recursions in quantum programming considered before, the control flows of this kind of quantum recursions are quantum rather than classical. 
Interestingly, this notion of quantum recursion enables us to construct a new class of quantum walks, called recursive quantum walks, whose behaviours seem very different from the quantum walks defined in the previous literatures. 

In this paper, we define the operational and denotational semantics of quantum recursions. The equivalence between these two semantics of quantum recursions are established. Obviously, how to define the semantics of quantum recursions is a question that can be asked within the basic framework of quantum mechanics. But surprisingly, answering it requires the mathematical tools from second quantisation \cite{MR04} because variable number of identical particles are employed to implement the quantum \textquotedblleft coins\textquotedblright\ involved in the computation of a quantum recursion. 

This paper is organised as follows. To make the paper self-contained, in Section \ref{QCQC} we recall the notions of quantum case statement and quantum choice from \cite{YYF12}, \cite{YYF13}. The syntax of quantum recursive programs is defined in Section \ref{sec-syn}. To give the reader a clearer picture, we choose not to include quantum measurements in the declarations of quantum recursions in this paper. It seems that quantum measurements can be added by combining the ideas used in this paper and those in \cite{YYF12}, \cite{YYF13}, but the presentation will be very complicated. In Section \ref{Rewk}, recursive quantum walks are considered as an example for further motivating the notion of quantum recursion. In particular, it is carefully explained that a formal description of the behaviour of recursive quantum walks requires the second quantisation method - a mathematical framework in which we are able to depict quantum systems with variable number of particles. For convenience of the reader, the basics of second quantisation is briefly reviewed in Section \ref{SQQ}. 
We define the semantics of quantum recursions in two steps. The first step is carried out in Section \ref{SEMA} where we construct a domain of operators in the free Fock space and prove continuity of semantic functionals of quantum programs with procedure identifiers. Then recursive equations are solved in the free Fock space by introducing the creation functional and by employing the standard fixed point theorem. The second step is completed in Section \ref{Recover} where the solutions of recursive equations are symmetralised so that they can apply in the physically meaningful framework, namely the symmetric and antisymmetric Fock spaces of bosons and fermions. A special class of quantum recursions, namely quantum while-loops with quantum control flow are examined in Section \ref{QWLO}. A short conclusion is drawn in Section \ref{CON1} with several problems for further studies.

\section{Quantum Case Statement and Quantum Choice}\label{QCQC}

Case statement in classical programming languages is a very useful program construct for case analysis, see \cite {Di75} for example. A quantum extension of case statement was defined in terms of measurements in various quantum programming languages, for example, Sanders and Zuliani's qGCL \cite{SZ00}, \cite{Zu01} and Selinger's QPL \cite{Se04}. The author defined another quantum case statement using external quantum \textquotedblleft coin\textquotedblright\ and further introduced quantum choice as a variant of quantum case statement in \cite{YYF12}, \cite{YYF13}. In this section, we recall these two program constructs from \cite{YYF13}.

\subsection{Quantum \textquotedblleft if...then...else\textquotedblright}
Let us start from the simplest case - a quantum counterpart of \textquotedblleft if...then...else\textquotedblright. Assume that $c$ is a qubit of which the state Hilbert space $\mathcal{H}_c$ has $|0\rangle$, $|1\rangle$ as an orthonormal basis. Furthermore, assume that $U_0$ and $U_1$ are two unitary transformations acting on a quantum system $q$ of which the state Hilbert space is $\mathcal{H}_q$. The system $q$ is called the principal quantum system. The action of $U_0$ on system $q$ can be thought of as a quantum program and is denoted $U_0[q]$. Similarly, we write $U_1[q]$ for the action of $U_1$ on $q$. Then a kind of quantum \textquotedblleft if...then...else\textquotedblright\ can be defined by employing qubit $c$ as a \textquotedblleft quantum coin\textquotedblright, and it is written as: \begin{equation}\label{qc1}\begin{split}&\mathbf{qif}\ [c]\ |0\rangle\rightarrow U_0[q]\\ &\ \ \ \ \ \ \ \ \square\ |1\rangle\rightarrow U_1[q]\\ &\mathbf{fiq}\end{split}\end{equation} in a way similar to Dijktra's guarded commands \cite{Di75}. The semantics of statement (\ref{qc1}) is an unitary operator $$U\stackrel{\triangle}{=}(|0\rangle\rightarrow U_0)\square (|1\rangle\rightarrow U_1)$$ on the tensor product $\mathcal{H}_c\otimes\mathcal{H}_q$ (i.e. the state Hilbert space of the composed system of \textquotedblleft coin\textquotedblright\ $c$ and principal system $q$) defined by  
\begin{equation*}U|0,\psi\rangle=|0\rangle U_0|\psi\rangle,\ \ \ \ \ U|1,\psi\rangle=|1\rangle U_1|\psi\rangle\end{equation*} for any $|\psi\rangle$ in $\mathcal{H}_q$. It can be represented by the following diagonal matrix \begin{equation*}U=|0\rangle\langle 0|\otimes U_0+|1\rangle\langle 1|\otimes U_1=\left(\begin{array}{cc}U_0 & 0\\ 0 & U_1\end{array}\right).\end{equation*} We call $U$ the guarded composition of $U_0$ and $U_1$ along with basis $|0\rangle,|1\rangle$. Moreover, let $V$ be a unitary operator in the state Hilbert space $\mathcal{H}_c$ of the \textquotedblleft coin\textquotedblright\ $c$. The action of $V$ on $c$ can also be thought of as a program and is denoted $V[c]$. Then the quantum choice of $U_0[q]$ and $U_1[q]$ with \textquotedblleft coin-tossing\textquotedblright\ $V[c]$ is defined to be \begin{equation}\label{qc2}\begin{split}V[c];\ &\mathbf{qif}\ [c]\ |0\rangle\rightarrow U_0[q]\\ &\ \ \ \ \ \ \ \ \square\ |1\rangle\rightarrow U_1[q]\\ &\mathbf{fiq}\end{split}\end{equation} Here and in the sequel, $P;Q$ denotes the sequential composition of programs $P$ and $Q$; that is, program $P$ followed by program $Q$. Using a notation similar to probabilistic choice in a probabilistic programming language like pGCL \cite{MM05}, program (\ref{qc2}) can be written as \begin{equation}\label{qc3}U_0[q]\oplus_{V[c]} U_1[q]\end{equation} Obviously, the semantics of quantum choice (\ref{qc3}) is the unitary matrix $U(V\otimes I_q),$ where $I_q$ is the identity operator in $\mathcal{H}_q$.

The idea of defining quantum \textquotedblleft if...then...else\textquotedblright\ using \textquotedblleft quantum coin\textquotedblright\ was actually borrowed from quantum walks. Here, let us consider the one-dimensional quantum walks \cite{AB01} as an example.

\begin{exam}\label{ex1} The simplest random walk is the one-dimensional walk in which a particle moves on a lattice marked by integers $\mathbb{Z}$, and at each step it moves one position left or right, depending on the flip of a fair coin. The Hadamard walk is a quantum variant of the one-dimensional random walk. Its state Hilbert space is $\mathcal{H}_d\otimes\mathcal{H}_p$, where $\mathcal{H}_d=\spa \{|L\rangle,|R\rangle\}$, $L,R$ are used to indicate the direction Left and Right, respectively, $\mathcal{H}_p=\spa\{|n\rangle: n\in\mathbb{Z}\}$, and $n$ indicates the position marked by integer $n$. One step of the Hadamard walk is represented by the unitary operator $W=T(H\otimes I)$, where the translation $T$ is a unitary operator in $\mathcal{H}_d\otimes\mathcal{H}_p$ defined by \begin{equation*}T|L, n\rangle=|L, n-1\rangle,\ \ \ \ \ T|R, n\rangle=|R, n+1\rangle\end{equation*} for every $n\in\mathbb{Z}$, $$H=\frac{1}{\sqrt{2}}\left(\begin{array}{cc}1 & 1\\ 1 & -1\end{array}\right)$$ is the Hadamard transform in the direction space $\mathcal{H}_d$, and $I$ is the identity operator in the position space $\mathcal{H}_p$. The Hadamard walk is described by repeated applications of operator $W$.

Now let us see how the idea of quantum case statement and quantum choice disguises in the construction of the Hadamard walk. If we define the left and right translation operators $T_L$ and $T_R$ in the position space $\mathcal{H}_p$ by $$T_L|n\rangle=|n-1\rangle,\ \ \ \ \ T_R|n\rangle=|n+1\rangle$$ for each $n\in\mathbb{Z}$, then the translation operator $T$ can be broken into a quantum case statement of $T_L$ and $T_R$:
\begin{equation}\label{tran}\begin{split}& T= \mathbf{qif}\ [d]\ |L\rangle\rightarrow T_L[p]\\ &\ \ \ \ \ \ \ \ \ \ \ \ \ \ \ \ \square\ |R\rangle\rightarrow T_R[p]\\ &\ \ \ \ \ \ \ \mathbf{fiq}\end{split}\end{equation}
where $d$ is a \textquotedblleft direction coin\textquotedblright, and $p$ is a variable used to denote the position. Furthermore, the single-step walk operator $W$ can be seen as the quantum choice $$T_L[p]\oplus_{H[d]}T_R[p].$$
 \end{exam}

Recently, physicists have been very interested in implementing quantum control for unknown subroutines \cite{ZRK11}, \cite{CDP13}, \cite{FDD14}, which is essentially a quantum \textquotedblleft if...then...else\textquotedblright or a quantum choice.

\subsection{Quantum Case Statement and Quantum Choice with Multiple Branches}

We now generalise the quantum case statement (\ref{qc1}) and quantum choice (\ref{qc2}) to the case with more than two branches. Let $n\geq 2$ and $c$ denote an $n-$level quantum system with state Hilbert space $\mathcal{H}_c=\spa\{|0\rangle,|1\rangle,...,|n-1\rangle\}$. For each $0\leq i<n$, let $U_i$ be a unitary operator or \textit{the zero operator} in the state Hilbert space $\mathcal{H}_q$ of the principal system $q$. Using system $c$ as a \textquotedblleft quantum coin\textquotedblright, we can define a quantum case statement:
\begin{equation}\label{qc4}\begin{split}&\mathbf{qif}\ [c]\ (\square i\cdot |i\rangle\rightarrow U_i[q])\ \mathbf{qif} = \mathbf{qif}\ [c]\ |0\rangle\rightarrow U_0[q]\\ &\ \ \ \ \ \ \ \ \ \ \ \ \ \ \ \ \ \ \ \ \ \ \ \ \ \ \ \ \ \ \ \ \ \ \ \ \ \ \ \ \ \ \ \ \ \ \ \ \ \ \ \ \ \ \ \ \ \ \ \ \square\ |1\rangle\rightarrow U_1[q]\\ &\ \ \ \ \ \ \ \ \ \ \ \ \ \ \ \ \ \ \ \ \ \ \ \ \ \ \ \ \ \ \ \ \ \ \ \ \ \ \ \ \ \ \ \ \ \ \ \ \ \ \ \ \ \ \ \ \ \ \ \ \ \ \ \ \ \ \ \ \ ..........\\ &\ \ \ \ \ \ \ \ \ \ \ \ \ \ \ \ \ \ \ \ \ \ \ \ \ \ \ \ \ \ \ \ \ \ \ \ \ \ \ \ \ \ \ \ \ \ \ \ \ \ \ \ \ \ \ \ \ \ \ \ \square\ |n-1\rangle\rightarrow U_{n-1}[q]\\ & \ \ \ \ \ \ \ \ \ \ \ \ \ \ \ \ \ \ \ \ \ \ \ \ \ \ \ \ \ \ \ \ \ \ \ \ \ \ \ \ \ \ \ \ \ \ \ \ \ \ \ \ \ \ \mathbf{fiq}\end{split}\end{equation} The reason for allowing some of $U_i$'s being the zero operator is that if $U_i[q]$ is a program containing recursion then it may not terminate. In the case that $U_i$ is the zero operator, we usually drop of the $i$th branch of the statement (\ref{qc4}). Furthermore, let $V$ be a unitary operator in the \textquotedblleft coin\textquotedblright\ space $\mathcal{H}_c$. Then we can define a quantum choice:
\begin{equation}\label{qc5}V[c]\ (\bigoplus_i |i\rangle\rightarrow U_i[q]) = V[c]; \mathbf{qif}\ [c]\ (\square i\cdot |i\rangle\rightarrow U_i[q])\ \mathbf{qif}\end{equation} The semantics of quantum case statement (\ref{qc4}) is the unitary operator $$U\stackrel{\triangle}{=}\square \left(c, |i\rangle\rightarrow U_i\right)$$ in $\mathcal{H}_c\otimes\mathcal{H}_q$ defined by
$U|i,\psi\rangle=|i\rangle U_i|\psi\rangle$ for any $0\leq i<n$ and $|\psi\rangle$ in $\mathcal{H}_q$. The operator $U$ is called the guarded composition of $U_i$'s along with basis $\{|i\rangle\}$. It is represented by the diagonal matrix \begin{equation}\label{qccq}U=\sum_{i=0}^{n-1}\left(|i\rangle_c\langle i|\otimes U_i\right)={\rm diag} (U_0,U_1,...,U_{n-1})=\left(\begin{array}{cccc}U_0 & & & \mathbf{0}\\ & U_1 & & \\ & & ... & \\ \mathbf{0}& & & U_{n-1}\end{array}\right)\end{equation} The semantics of quantum choice (\ref{qc5}) is then the operator $U(V\otimes I_q)$, where $I_q$ is the identity operator in $\mathcal{H}_q$.

Quantum walks on a graph \cite{AA01} can be conveniently expressed in terms of the above generalised quantum case statement and choice, as shown in the following:

\begin{exam}\label{ex2} A random walk on a directed graph $G=(V,E)$ is described by repeated applications of stochastic matrix $P=(P_{uv})_{u,v\in V}$, where $$P_{uv}=\begin{cases}\frac{1}{d_u} &{\rm if}\ (u,v)\in E,\\ 0 &{\rm otherwise}\end{cases}$$ where $d_u$ is the outgoing degree of $u$, i.e. the number of edges outgoing from $u$. In particular, if $G$ is $d-$regular, i.e. all nodes have the same degree $d$, then $P_{uv}=\frac{1}{d}$ for all $u,v\in V$. A quantum walk on graph $G$ is a quantum counterpart of the random walk. Let $\mathcal{H}_V=\spa \{|v\rangle:v\in V\}$ be the Hilbert space spanned by states corresponding to the vertices in $G$. We now assume that $G$ is $d-$regular. Then each edge in $G$ can be labelled by a number among $1,2,...,d$ so that for any $1\leq a\leq d$, the edges labelled $a$ form a permutation. Let $\mathcal{H}_A=\spa\{|1\rangle,|2\rangle,...,|d\rangle\}$ be an auxiliary Hilbert space of dimension $d$, called the \textquotedblleft coin space\textquotedblright. The shift operator $S$ is defined in $\mathcal{H}_A\otimes\mathcal{H}_V$ by $$S|a,v\rangle=|a,v_a\rangle$$ for $1\leq a\leq d$ and $v\in V$, where $v_a$ is the $a-$th neighbour of $v$, i.e. the vertex reached from $v$ through the outgoing edge labelled $a$. Furthermore, let $C$ be a unitary operator in $\mathcal{H}_A$, called the \textquotedblleft coin-tossing operator\textquotedblright. Then one step of the quantum walk is modelled by the operator $W=S(C\otimes I)$, where $I$ is the identity operator in $\mathcal{H}_V$. The quantum walk is described by repeated applications of $W$.

If for each $1\leq a\leq d$, we define the $a-$th shift operator $S_a$ in $\mathcal{H}_V$ by $$S_a|v\rangle=|v_a\rangle$$ for any $v\in V$, then the shift operator $S$ can be seen as a quantum case statement:
\begin{equation*}\begin{split} S=\ &\mathbf{qif}\ [c]\ (\square a\cdot |a\rangle\rightarrow S_a[q])\ \mathbf{qif}\\ =\ &\mathbf{qif}\ [c]\ |1\rangle\rightarrow S_1[q]\\ &\ \ \ \ \ \ \ \ \square\ |2\rangle\rightarrow S_2[q]\\ &\ \ \ \ \ \ \ \ \ \ \ \ \ ..........\\ &\ \ \ \ \ \ \ \ \square\ |d\rangle\rightarrow S_{d}[q]\\ & \mathbf{fiq}\end{split}\end{equation*}
where $c$ and $q$ are two variables denoting quantum systems with state spaces $\mathcal{H}_A$ and $\mathcal{H}_V$, respectively. Consequently, the single-step walk operator $W$ is the quantum choice: $$W=\ C[c] (\bigoplus_a |a\rangle\rightarrow S_a[q])$$
\end{exam}

The quantum case statement (\ref{qc4}) and quantum choice (\ref{qc5}) can be further generalised to the case where unitary transformations $U_0[q], U_1[q],...,U_{n-1}[q]$ are replaced by general quantum programs that may contain quantum measurements. It is quite involved to define the semantics of such general quantum case statement and choice; for details we refer to \cite{YYF12}, \cite{YYF13}.

\section{Syntax of Quantum Recursive Programs}\label{sec-syn}

A new notion of quantum recursion with quantum control flow can be defined based on quantum case statement and quantum choice discussed in the last section. In this short section, we formally define the syntax of quantum recursive programs.   

We assume two sets of quantum variables: principal system variables, ranged over by $p, q,...$, and \textquotedblleft coin\textquotedblright\ variables, ranged over by $c,d,...$. These two sets are required to be disjoint. We also assume a set of procedure identifiers, ranged over by $X,X_1,X_2,...$. Then program schemes are defined by the following syntax:
\begin{equation*}P::=\ X\ |\ \mathbf{abort}\ |\ \mathbf{skip}\ |\ P_1;P_2\ |\ U[\overline{c}, \overline{q}]\ |\ \mathbf{qif}\ [c] (\square i\cdot |i\rangle\rightarrow P_i)\ \mathbf{fiq}\end{equation*}
where:\begin{itemize}\item $X$ is a procedure identifier; programs $\mathbf{abort}$ and $\mathbf{skip}$ are the same as in a classical programming language; $P_1;P_2$ is the sequential composition of $P_1$ and $P_2$. \item In unitary transformation $U[\overline{c}, \overline{q}]$, $\overline{c}$ a sequence of \textquotedblleft coin\textquotedblright\ variables, $\overline{q}$ is a sequence of principal system variables, and $U$ is a unitary operator in the state Hilbert space of the system consisting of $\overline{c}$ and $\overline{q}$. We will always put \textquotedblleft coin\textquotedblright\ variables before principal system variables. Both of $\overline{c}$ and $\overline{q}$ are allowed to be empty. When $\overline{c}$ is empty, we simply write $U[\overline{q}]$ for $U[\overline{c}, \overline{q}]$ and it describes the evolution of the principal system $\overline{q}$; when $\overline{q}$ is empty, we simply $U[\overline{c}]$ for $U[\overline{c}, \overline{q}]$ and it describes the evolution of the \textquotedblleft coins\textquotedblright\ $\overline{c}$. If both $\overline{c}$ and $\overline{q}$ are not empty, then $U[\overline{c}, \overline{q}]$ describes the interaction between  \textquotedblleft coins\textquotedblright\ $\overline{c}$ the principal system $\overline{q}$.
\item In quantum case statement $\mathbf{qif}\ [c] (\square i\cdot |i\rangle\rightarrow P_i)\ \mathbf{fiq}$, $c$ is a \textquotedblleft coin\textquotedblright\ variable, and $\{|i\rangle\}$ is an orthonormal basis of the state Hilbert space of $c$. It is required not to occur in all the subprograms $P_i$'s because according to its physical interpretation, a \textquotedblleft coin\textquotedblright\ is always external to the principal system. This program construct is a generalisation of equation (\ref{qc4}). \end{itemize} 

As a generalisation of equation (\ref{qc5}), the program construct of quantum choice can be defined in terms of quantum case statement and sequential composition.
Let $P$ be a program contains only \textquotedblleft coin\textquotedblright\ $c$, let $\{|i\rangle\}$ be an orthonormal basis of the state Hilbert space of $c$, and let $P_i$ be a program for each $i$ . Then the quantum choice of $P_i$'s according to $P$ along the basis $\{|i\rangle\}$ is defined as \begin{equation}\label{quch}[P(c)]\bigoplus_{i}\left(|i\rangle \rightarrow P_i\right)\stackrel{\triangle}{=}P; \mathbf{qif}\ [c] \left(\square i\cdot |i\rangle \rightarrow P_i\right)\ \mathbf{fiq}.\end{equation}
If the \textquotedblleft coin\textquotedblright\ is a qubit, then quantum choice (\ref{quch}) can be abbreviated as $$P_0\oplus_{P}P_1.$$ Intuitively, quantum choice (\ref{quch}) runs a \textquotedblleft coin-tossing\textquotedblright\ subprogram $P$ followed by an alternation of a family of subprograms $P_0,P_1,....$ The \textquotedblleft coin-tossing\textquotedblright\ subprogram $P$ creates a superposition of the execution paths of $P_0,P_1,...,$ and during the execution of the alternation, each $P_i$ is running along its own path, but the whole program is executed in a  superposition of execution paths of $P_0,P_1,....$ This picture can be imaginatively termed as the \textit{superposition-of-programs} paradigm. 

The semantics of quantum programs without procedure identifiers (and thus without recursion) can be easily defined. The principal system of a quantum program $P$ is the composition of the systems denoted by principal variables appearing in $P$. We write $\mathcal{H}$ for the state Hilbert space of the principal system. 

\begin{defn}\label{seman-w}The semantics $\llbracket P\rrbracket$ of a program $P$ without procedure identifiers is inductively defined as follows:\begin{enumerate}\item If $P=\mathbf{abort}$, then $\llbracket P\rrbracket = 0$ (the zero operator in $\mathcal{H}$), and if $P=\mathbf{skip}$, then $\llbracket P\rrbracket=I$ (the identity operator in $\mathcal{H}$);
\item If $P$ is an unitary transformation $U[\overline{c}, \overline{q}]$, then $\llbracket P\rrbracket$ is the unitary operator $U$ (in the state Hilbert space of the system consisting of  $\overline{c}$ and $\overline{q}$);\item If $P=P_1;P_2$, then $\llbracket P\rrbracket=\llbracket P_2\rrbracket\cdot \llbracket P_1\rrbracket$;\item If $P=\mathbf{qif}\ [c](\square i\cdot|i\rangle\rightarrow P_i)\ \mathbf{fiq}$, then $$\llbracket P\rrbracket =\square (c, |i\rangle\rightarrow \llbracket P_i\rrbracket)\stackrel{\triangle}{=}\sum_i\left(|i\rangle_c\langle i|\otimes \llbracket P_i\rrbracket\right)$$ (see equation (\ref{qccq}) for the special case of unitary operators). \end{enumerate}
\end{defn}

Finally, we can define the syntax of quantum recursive programs. If a program scheme contains at most the procedure identifiers $X_1,...,X_m$, then we write $P=P[X_1,...,X_m].$ 

\begin{defn}\label{RecP-DF}\begin{enumerate}\item Let $X_1,...,X_m$ be different procedure identifiers. A declaration for $X_1,...,X_m$ is a system of equations:  
\begin{equation*}D:\ \begin{cases}X_1\Leftarrow P_1,\\ \ \ \ \ \ \ ......\\ X_m\Leftarrow P_m,\end{cases}\end{equation*} where for every $1\leq i\leq m$, $P_i=P_i[X_1,...,X_m]$ is a program scheme containing at most procedure identifiers $X_1,...,X_m$. 
 
\item A recursive program consists of a program scheme $P=P[X_1,...,X_m]$, called the main statement, and a declaration $D$ for $X_1,...,X_m$ such that all \textquotedblleft coin\textquotedblright\ variables in $P$ do not appear in $D$; that is, they do not appear in the procedure bodies $P_1,...,P_m$.\end{enumerate}\end{defn}

The requirement in the above definition that the \textquotedblleft coins\textquotedblright\ in the main statement $P$ and those in the declaration $D$ are distinct is obviously necessary because a \textquotedblleft coin\textquotedblright\ used to define a quantum case statement is always considered to be external to its principal system.  

Now the question is: how to define the semantics of quantum recursive programs? 

\section{Motivating Example: Recursive Quantum Walks}\label{Rewk}

As a motivating example of quantum recursive program, let us consider a variant of quantum walks, called recursive quantum walks. For simplicity, we focus on the recursive Hadamard walk - a modification of Example \ref{ex1}. Recursive quantum walks on a graph can be defined by modifying Example \ref{ex2} in a similar way.

\subsection{Specification of Recursive Quantum Walks}

Recall that the single-step operator $W$ of the Hadamard walk is a quantum choice, which is the sequential composition of a \textquotedblleft coin-tossing\textquotedblright\ Hadamard operator $H$ on the \textquotedblleft direction coin\textquotedblright\ $d$ and translation operator $T$ on the position variable $p$. The translation $T[p]$ is a quantum case statement that selects left or right translations according to the basis states $|L\rangle, |R\rangle$ of the \textquotedblleft coin\textquotedblright\ $d$. If $d$ is in state $|L\rangle$ then the walker moves one position left, and if $d$ is in state $|R\rangle$ then it moves one position right. An essential difference between a random walk and a quantum walk is that the \textquotedblleft coin\textquotedblright\ of the latter can be in a superposition of the basis states $|L\rangle, |R\rangle$, and thus a superposition of left and right translations $T_L[p]$ and $T_R[p]$ is created. The Hadamard walk is then defined in a simple way of recursion with the single-step operator $W$, namely repeated applications of $W$.
Now we modify slightly the Hadamard walk using a little bit more complicated recursion.

\begin{exam}\label{ex3} \begin{enumerate}\item The unidirectionally recursive Hadamard walk first runs the \textquotedblleft coin-tossing\textquotedblright\ Hadamard operator $H[d]$ and then a quantum case statement:
if the \textquotedblleft direction coin\textquotedblright\ $d$ is in state $|L\rangle$ then the walker moves one position left, and if $d$ is in state $|R\rangle$ then it moves one position right, followed by \textbf{a procedure behaving as the recursive walk itself}. In the terminology of programming languages, the recursive Hadamard walk is defined to a program $X$ declared by the following recursive equation:
\begin{equation}\label{rhw}X\Leftarrow\  T_L[p] \oplus_{H[d]} (T_R[p];X)\end{equation} where $d, p$ are the direction and position variables, respectively.
\item The bidirectionally recursive Hadamard walk first runs the \textquotedblleft coin-tossing\textquotedblright\ Hadamard operator $H[d]$ and then a quantum case statement: if the \textquotedblleft direction coin\textquotedblright\ $d$ is in state $|L\rangle$ then the walker moves one position left, followed by \textbf{a procedure behaving as the recursive walk itself}, and if $d$ is in state $|R\rangle$ then it moves one position right, also followed by \textbf{a procedure behaving as the recursive walk itself}. More precisely, the walk can be defined to be the program $X$ declared by the following two recursive equations:
\begin{equation}\label{drhw}X\Leftarrow (T_L[p];X)\oplus_{H[d]}(T_R[p];X).\end{equation}
\item A variant of the bidirectionally recursive Hadamard walk is the program $X$ (or $Y$) declared by the following system of recursive equations:
\begin{equation}\label{ddrhw}\begin{cases}X\Leftarrow T_L[p]\oplus_{H[d]} (T_R[p]; Y),\\ Y\Leftarrow (T_L[p];X)\oplus_{H[d]} T_R[p].\end{cases}\end{equation}
\item Note that we used the same \textquotedblleft coin\textquotedblright\ $d$ in the two equations of (\ref{ddrhw}). If two different \textquotedblleft coins\textquotedblright\ $d$ and $e$ are used, then we have another variant of the bidirectionally recursive Hadamard walk specified by 
\begin{equation*}\begin{cases}X\Leftarrow T_L[p]\oplus_{H[d]} (T_R[p]; Y),\\ Y\Leftarrow (T_L[p];X)\oplus_{H[e]} T_R[p].\end{cases}\end{equation*}
\item We can define a recursive quantum walk in another way if quantum case statement with three branches is employed: 
\begin{align*}X\Leftarrow U[d];\ &\mathbf{qif}\ [d]\ |L\rangle\rightarrow T_L[p]\\
&\ \ \ \ \ \ \ \ \square\ |R\rangle\rightarrow T_R[p]\\
&\ \ \ \ \ \ \ \ \square\ |I\rangle\rightarrow X\\
&\mathbf{fiq}
\end{align*} where $d$ is not a qubit but a qutrit with state space $\mathcal{H}_d=\spa\{|L\rangle, |R\rangle, |I\rangle\}$, $L, R$ stand for the directions Left and Right, respectively, and $I$ for Iteration, and $U$ is a $3\times 3$ unitary matrix, e.g. the $3-$dimensional Fourier transform:
$$F_3=\left(\begin{array}{ccc}1&1&1\\ 1& e^{\frac{2}{3}\pi i} & e^{\frac{4}{3}\pi i}\\ 1 & e^{\frac{4}{3}\pi i}&e^{\frac{2}{3}\pi i}\end{array}\right).$$
\end{enumerate}

Now let us have a glimpse of the behaviours of recursive quantum walks. We use $E$ to denote the empty program or termination. A configuration is defined to be a pair $(P,|\psi\rangle)$ with $P$ being a program or the empty program $E$, and $|\psi\rangle$ a pure state of the quantum system. Then the  behaviour of a program can be visualised by a sequence of transitions between superpositions of configurations. Here, we only consider the unidirectionally recursive quantum walk $X$ declared by equation (\ref{rhw}). Assume that it is initialised in state $|L\rangle_d|0\rangle_p$; that is, the \textquotedblleft coin\textquotedblright\ is in direction $L$ and the walker is at position $0$. Then we have:
 \begin{equation}\label{bdhw}\begin{split}(X,|L\rangle_d|0\rangle_p)&\rightarrow \frac{1}{\sqrt{2}}(E,|L\rangle_d|-1\rangle_p)+\frac{1}{\sqrt{2}}(X,|R\rangle_d|1\rangle_p)\\
 &\rightarrow \frac{1}{\sqrt{2}}(E,|L\rangle_d|-1\rangle_p)+\frac{1}{2}(E,|R\rangle_d|L\rangle_{d_1}|0\rangle_p)+\frac{1}{2}(X,|R\rangle_d|R\rangle_{d_1}|2\rangle_p)
 \\& \rightarrow ......
 \\
 &\rightarrow \sum_{i=0}^{n}\frac{1}{\sqrt{2^{i+1}}}(E,|R\rangle_{d_0}...|R\rangle_{d_{i-1}}|L\rangle_{d_i}|i-1\rangle_p)\\ &\ \ \ \ \ \ \ \ \ \ \ \ \ \ \ \ \ \ \ \ \ \ \ \ \ \ \ +\frac{1}{\sqrt{2^{n+1}}}(X,|R\rangle_{d_0}...
 |R\rangle_{d_{n-1}}|R\rangle_{d_n}|n+1\rangle_p)
 \end{split}\end{equation} Here, $d_0=d$, and new quantum \textquotedblleft coins\textquotedblright\ $d_1,d_2,...$ that are identical to the original \textquotedblleft coin\textquotedblright\ $d$ are introduced in order to avoid the conflict of variables for \textquotedblleft coins\textquotedblright. \end{exam}

The above recursive quantum walks are good examples of quantum recursion, but their behaviours are not very interesting. It has been well-understood that the major difference between the behaviours of classical random walks and quantum walks is caused by quantum interference - two separate paths leading to the same point may be out of phase and cancel one another \cite{AB01}. It is clear from equation (\ref{bdhw}) that quantum interference does not happen in the unidirectionally recursive quantum walk. Similarly, no quantum interference occurs in the bidirectionally recursive quantum walks defined in the above example. The following is a much more interesting recursive quantum walk that shows a new phenomenon of quantum interference: the paths that are cancelled in a quantum walk are finite. However, it is possible that infinite paths are cancelled in a recursive quantum walk.

 \begin{exam} Let $n\geq 2$. A variant of unidirectionally recursive quantum walk can be defined as the program $X$ declared by the following recursive equation:
 \begin{equation}\label{qintw}X\Leftarrow (T_L[p]\oplus_{H[d]} T_R[p])^n; ((T_L[p];X)\oplus_{H[d]} (T_R[p];X))\end{equation}
 Here, we use $P^n$ to denote the sequential composition of $n$ $P$'s. Now let us look at the behaviour of this walk. We assume that the walk is initialised in state $|L\rangle_d|0\rangle_p$. Then the first three steps of the walk are given as follows:
 \begin{equation}\label{qintw-0}\begin{split}&(X,|L\rangle_d|0\rangle_p)\rightarrow \frac{1}{\sqrt{2}}[(X_1,|L\rangle_d|-1\rangle_p+(X_1,|R\rangle_d|1\rangle_p)]\\
 &\rightarrow \frac{1}{2}[(X_2,|L\rangle_d|-2\rangle_p)+(X_2,|R\rangle_d|0\rangle_p)+(X_2,|L\rangle_d|0\rangle_p)-(X_2,|R\rangle_d|2\rangle_p)]\\
 &\rightarrow \frac{1}{2\sqrt{2}}[(X_3,|L\rangle_d|-3\rangle_p)+(X_3,|R\rangle_d|-1\rangle_p)+(X_3,|L\rangle_d|-1\rangle_p)-(X_3,|R\rangle_d|1\rangle_p)
 \\ &\ \ \ \ \ \ \ \ \ \ \ \ \ \ \ \ \ +(X_3,|L\rangle_d|-1\rangle_p)+(X_3,|R\rangle_d|1\rangle_p)-(X_3,|L\rangle_d|1\rangle_p)+(X_3,|R\rangle_d|3\rangle_p)]\\
 &=\frac{1}{2\sqrt{2}}[(X_3,|L\rangle_d|-3\rangle_p)+(X_3,|R\rangle_d|-1\rangle_p)+2(X_3,|L\rangle_d|-1\rangle_p)\\ &\ \ \ \ \ \ \ \ \ \ \ \ \ \ \ \ \ \ \ \ \ \ \ \ \ \ \ \ \ \ \ \ \ \ \ \ \ \ \ \ \ \ \ \ \ \ \ \ \ \ \ \ \ \ \ \ \ \ \ -(X_3,|L\rangle_d|1\rangle_p)+(X_3,|R\rangle_d|3\rangle_p)]
 \end{split}
 \end{equation} where $$X_i=(T_L[p]\oplus_{H[d]} T_R[p])^{n-i}; ((T_L[p];X)\oplus_{H[d]} (T_R[p];X))$$ for $i=1,2,3$. We observe that in the last step of equation (\ref{qintw-0}) two configurations $-(X,|R\rangle_d|1\rangle_p)$ and $(X,|R\rangle_d|1\rangle_p)$ cancel one another in the last part of the above equation. It is clear that both of them can generate infinite paths because they contain the recursive walk $X$ itself.
 
The behaviour of the recursive program specified by the following equation: \begin{equation}\label{qintw-1}X\Leftarrow (T_L[p]\oplus_{H[d]} T_R[p])^n; ((T_L[p];X)\oplus_{H[d]} (T_R[p];X))\end{equation}
is even more puzzling. Note that equation (\ref{qintw-1}) is obtained from equation (\ref{qintw}) by changing the order of the two subprograms in its right-hand side. \end{exam}

\subsection{How to solve recursive quantum equations?}

We have already seen the first steps of the recursive quantum walks. But a precise description of their behaviours amounts to solving recursive equations (\ref{rhw}), (\ref{drhw}), (\ref{ddrhw}) and (\ref{qintw}). In mathematics, a standard method for finding the least solution to  an equation $x= f(x)$ with $f$ being a function from a lattice into itself is as follows: let $x_0$ be the least element of the lattice. We take the iterations of $f$ starting from $x_0$: \begin{equation*}\begin{cases}x^{(0)}=x_0,\\ x^{(n+1)}=f(x^{(n)})\ {\rm for}\ n\geq 0.\end{cases}\end{equation*} If $f$ is monotone and the lattice is complete, then the limit $\lim_{n\rightarrow\infty}x^{(n)}$ of iterations exists; and furthermore if $f$ is continuous, then this limit is the least solution of the equation. In the theory of programming languages \cite{ABO09}, a syntactic variant of this method is employed to define the semantics of a recursive program declared by, say, equation $X\Leftarrow F(X)$, where $F(\cdot)$ is presented in a syntactic rather than semantic way: let
\begin{equation*}\begin{cases}X^{(0)}=\mathbf{Abort},\\ X^{(n+1)}=F[X^{(n)}/X]\ {\rm for}\ n\geq 0.\end{cases}\end{equation*}
where $F[X^{(n)}/X]$ is the result of substitution of $X$ in $F(X)$ by $X^{(n)}$. The program $X^{(n)}$ is called the $n$th syntactic approximation of $X$.
Roughly speaking, the syntactic approximations $X^{(n)}$ $(n=0,1,2,...)$ describe the initial fragments of the behaviour of the recursive program $X$.
Then the semantics $\llbracket X\rrbracket$ of $X$ is defined to be the limit of the semantics $\llbracket X^{(n)}\rrbracket$ of its syntactic approximations $X^{(n)}$: $$\llbracket X\rrbracket=\lim_{n\rightarrow\infty}\llbracket X^{(n)}\rrbracket.$$ Now we apply this method to the unidirectionally recursive Hadamard walk and construct its syntactic approximations as follows:
\begin{equation}\label{sap}\begin{split}&X^{(0)}=\mathbf{abort},\\ &X^{(1)}=T_L[p]\oplus_{H[d]} (T_R[p];\mathbf{abort}),\\ &X^{(2)}=T_L[p]\oplus_{H[d]} (T_R[p];T_L[p]\oplus_{H[d_1]} (T_R[p];\mathbf{abort})),\\
&X^{(3)}=T_L[p]\oplus_{H[d]} (T_R[p];T_L[p]\oplus_{H[d_1]} (T_R[p];T_L[p]\oplus_{H[d_2]} (T_R[p];\mathbf{abort}))),\\ & ............
\end{split}\end{equation}
However, a problem arises in constructing these approximations: we have to continuously introduce new \textquotedblleft coin\textquotedblright\ variables in order to avoid variable conflict; that is, for every $n=1,2,...$, we introduce a new \textquotedblleft coin\textquotedblright\ variable $d_n$ in the $(n+1)$th syntactic approximation. Obviously, variables $d,d_1,d_2,...$ must denote identical particles. Moreover, the number of the \textquotedblleft coin\textquotedblright\ particles that are needed in running the recursive Hadamard walk is usually unknown beforehand because we do not know  when the walk terminates. It is clear that this problem appears only in the quantum case but not in the theory of classical programming languages because it is caused by employing an external \textquotedblleft coin\textquotedblright\ system in defining a quantum case statement. Therefore, a solution to this problem requires a mathematical framework in which we can deal with quantum systems where the number of particles of the same type - the \textquotedblleft coins\textquotedblright\ - may vary.

\section{Second Quantisation}\label{SQQ}

Fortunately, physicists had developed a formalism for describing quantum systems with variable particle number, namely second quantisation, more than eighty years ago. For convenience of the reader, we recall basics of the second quantum method in this section.
\subsection{Fock Spaces} Let $\mathcal{H}$ be the state Hilbert space of one particle.
For any $n\geq 1$, we write $\mathcal{H}^{\otimes n}$ for the $n-$fold tensor product of $\mathcal{H}$. If we introduce the vacuum state $|\mathbf{0}\rangle$, then the $0-$fold tensor product of $\mathcal{H}$ can be defined as the one-dimensional space $\mathcal{H}^{\otimes 0}=\spa\{|\mathbf{0}\rangle\}$. Furthermore, the free Fock space over $\mathcal{H}$ is defined to be the direct sum \cite{Att}: $$\mathcal{F}(\mathcal{H})=\bigoplus_{n=0}^{\infty}\mathcal{H}^{\otimes n}.$$

The principle of symmetrisation in quantum physics \cite{MR04} indicates that the states of $n$ identical particles are either completely symmetric or completely antisymmetric with respect to the permutations of the particles. These particles are called bosons in the symmetric case and fermions in the antisymmetric case.  For each permutation $\pi$ of $1,...,n$, we define the permutation operator $P_\pi$ in $\mathcal{H}^{\otimes n}$ by $$P_\pi|\psi_1\otimes ...\otimes\psi_n\rangle=|\psi_{\pi(1)}\otimes ...\otimes\psi_{\pi(n)}\rangle$$ for all $|\psi_1\rangle,...,|\psi_n\rangle$ in $\mathcal{H}$. Furthermore, we define the symmetrisation and antisymmetrisation operators in $\mathcal{H}^{\otimes n}$ as follows: $$S_+=\frac{1}{n!}\sum_\pi P_\pi,\ \ \ \ \ \ S_-=\frac{1}{n!}\sum_{\pi}(-1)^\pi P_\pi$$ where $\pi$ ranges over all permutations of $1,...,n$, and $(-1)^\pi$ is the signature of the permutation $\pi$. For $v=+,-$ and any $|\psi_1\rangle,...,|\psi_n\rangle$ in $\mathcal{H}$, we write $$|\psi_1,...,\psi_n\rangle_v=S_v|\psi_1\otimes ...\otimes\psi_n\rangle.$$ Then the state space of $n$ bosons and that of fermions are $$\mathcal{H}_v^{\otimes n}=S_v\mathcal{H}^{\otimes n}=\spa\{|\psi_1,...,\psi_n\rangle_v:|\psi_1\rangle,...,|\psi_n\rangle\ {\rm are\ in}\ \mathcal{H}\}$$ for $v=+,-,$ respectively. If we set $\mathcal{H}_v^{\otimes 0}=\mathcal{H}^{\otimes 0}$, then the space of the states of variable particle number is the symmetric or antisymmetric Fock space: $$\mathcal{F}_v(\mathcal{H})=\bigoplus_{n=0}^\infty\mathcal{H}_v^{\otimes n}$$ where $v=+$ for bosons and $v=-$ for fermions. The elements of the Fock space $\mathcal{F}_v(\mathcal{H})$ (resp. the free Fock space $\mathcal{F}(\mathcal{H})$) are of the form $$|\Psi\rangle=\sum_{n=0}^\infty |\Psi(n)\rangle$$ with $|\Psi(n)\rangle\in\mathcal{H}_v^{\otimes n}$ (resp. $|\Psi(n)\rangle\in\mathcal{H}^{\otimes n}$) for $n=0,1,2,...$ and $\sum_{n=0}^\infty\langle\Psi(n)|\Psi(n)\rangle<\infty$.

\subsection{Operators in the Fock Spaces}

For each $n\geq 1$, let $\mathbf{A}(n)$ be an operator in $\mathcal{H}^{\otimes n}$. Then operator \begin{equation}\label{fockop}\mathbf{A}=\sum_{n=0}^\infty \mathbf{A}(n)\end{equation} is defined in the free Fock space $\mathcal{F}(\mathcal{H})$ as follows: 
\begin{equation*}\mathbf{A}\sum_{n=0}^\infty |\Psi(n)\rangle=\sum_{n=0}^\infty \mathbf{A}(n)|\Psi(n)\rangle\end{equation*} for any $|\Psi\rangle=\sum_{n=0}^\infty |\Psi(n)\rangle$ in $\mathcal{F}(\mathcal{H})$, where $\mathbf{A}(0)=0$; that is, the vacuum state is considered to be an eigenvector of operator $\mathbf{A}$ with eigenvalue $0$. 

If for each $n\geq 0$ and for each permutation $\pi$ of $1,...,n$, $P_\pi$ and $\mathbf{A}(n)$ commute; that is, $$P_\pi\mathbf{A}(n)=\mathbf{A}(n)P_\pi,$$ then operator $\mathbf{A}$ is said to be symmetric. A symmetric operator $\mathbf{A}=\sum_{n=0}^\infty \mathbf{A}(n)$ is an operator both in the symmetric Fock space $\mathcal{F}_+(\mathcal{H})$ and in the antisymmetric Fock space $\mathcal{F}_-(\mathcal{H})$: $\mathbf{A}(\mathcal{F}_v(\mathcal{H}))\subseteq\mathcal{F}_v(\mathcal{H})$ for $v=+,-$.
We can introduce the symmetrisation functional $\mathbb{S}$ that maps every operator $\mathbf{A}=\sum_{n=0}^\infty \mathbf{A}(n)$ to a symmetric operator:  
\begin{equation}\label{symz1}\mathbb{S}(\mathbf{A})=\sum_{n=0}^\infty \mathbb{S}(\mathbf{A}(n))\end{equation} where    
for each $n\geq 0$, \begin{equation}\label{symz2}\mathbb{S}(\mathbf{A}(n))=\frac{1}{n!}\sum_{\pi}P_\pi\mathbf{A}(n)P_\pi^{-1}\end{equation} with $\pi$ traversing over all permutations of $1,...,n$.

\subsubsection{Observables in the Fock Spaces} If for each $n\geq1$, the operator $\mathbf{A}(n)$ in equation (\ref{fockop}) is an observable of $n$ particles, then $\mathbf{A}$ is an extensive observable in the free Fock space $\mathbf{F}(\mathcal{H})$. In particular, let us consider one-body observables. Assume that $A$ is a single-particle observable. Then we can define one-body observable $$\mathbf{A}(n)=\sum_j A^{(n)}_j$$ for the system of $n$ particles, where $A^{(n)}_j=I^{\otimes (n-1)}\otimes A\otimes I^{\otimes (n-j)}$ (with $I$ being the identity operator in $\mathcal{H}$) is the action of $A$ on the $j$th factor of the tensor product $\mathcal{H}^{\otimes n}$; that is, 
$$A^{(n)}_j|\psi_1\otimes ... \otimes \psi_n\rangle=|\psi_1\otimes ...\otimes\psi_{j-1}\otimes A\psi_j\otimes \psi_{j+1}\otimes ...\otimes \psi_n\rangle$$ for all $|\psi_1\rangle,...|\psi_n\rangle$ in $\mathcal{H}$.
It is easy to see that $\mathbf{A}(n)$ commutes with the permutations: $$\mathbf{A}(n)|\psi_1, ..., \psi_n\rangle_v=\sum_{j=1}^n|\psi_1, ...,\psi_{j-1}, A\psi_j, \psi_{j+1}, ..., \psi_n\rangle_v.$$
Therefore, $\mathbf{A}=\sum_{n=0}^\infty\mathbf{A}(n)$ is symmetric. It is called a one-body observable in the Fock space $\mathcal{F}_v(\mathcal{H})$ for $v=+,-$. Similarly, we can define a $k-$body observable with $k\geq 2$.

\subsubsection{Evolutions in the Fock Spaces}
Let the (discrete-time) evolution of one particle is represented by unitary operator $U$. Then the evolution of $n$ particles without mutual interactions can be described by operator $\mathbf{U}(n)=U^{\otimes n}$ in $\mathcal{H}^{\otimes n}$: \begin{equation}\label{evo1}\mathbf{U}(n)|\psi_1\otimes ...\otimes\psi_n\rangle=|U\psi_1\otimes ...\otimes U\psi_n\rangle\end{equation} for all $|\psi_1\rangle,...,|\psi_n\rangle$ in $\mathcal{H}$. It is easy to verify that $\mathbf{U}(n)$ commutes with the permutations: $$\mathbf{U}(n)|\psi_1,...,\psi_n\rangle_v=|U\psi_1,...,U\psi_n\rangle_v.$$ So, the symmetric operator $\mathbf{U}=\sum_{n=0}^\infty\mathbf{U}(n)$ depicts the evolution of particles without mutual interactions in the Fock space $\mathcal{F}_v(\mathcal{H})$ for $v=+,-$.
\subsubsection{Creation and Annihilation of Particles}
The operator $\mathbf{U}$ defined by equation (\ref{fockop}) maps states of $n$ particles to states of particles of the same number.
The transitions between states of different particle numbers are described by the creation and annihilation operators. To each one-particle state $|\psi\rangle$ in $\mathcal{H}$, we associate the creation operator $a^\dag(\psi)$ in $\mathcal{F}_v(\mathcal{H})$ defined by $$a^\dag(\psi)|\psi_1,...,\psi_n\rangle_v=\sqrt{n+1}|\psi,\psi_1,...,\psi_n\rangle_v$$ for any $n\geq 0$ and all $|\psi_1\rangle,...,|\psi_n\rangle$ in $\mathcal{H}$. This operator adds a particle in the individual state $|\psi\rangle$ to the system of $n$ particles without modifying their respective states. The annihilation operator $a(\psi)$ is defined to be the Hermitian conjugate of $a^\dag(\psi)$, and it is not difficult to show that \begin{equation*}\begin{split}
a(\psi)|0\rangle&=0,\\
a(\psi)|\psi_1,...,\psi_n\rangle_v&=\frac{1}{\sqrt{n}}\sum_{i=1}^n (v)^{i-1}\langle\psi|\psi_i\rangle |\psi_1,...,\psi_{i-1},\psi_{i+1},...,\psi_n\rangle_v
\end{split}\end{equation*} Intuitively, operator $a(\psi)$ decreases the number of particles by one unit, while preserving the symmetry of the state.

\section{Solving Recursive Equations in the Free Fock Space}\label{SEMA}

Second quantisation provides us with the necessary tool for defining the semantics of quantum recursions. We first show how to solve recursive equations in the free Fock spaces without considering symmetry or antisymmetry of the particles that are used to implement the quantum \textquotedblleft coins\textquotedblright.  

\subsection{A Domain of Operators in the Free Fock Space}\label{domain}
Let $C$ be a set of quantum \textquotedblleft coins\textquotedblright. For each $c\in C$, let $\mathcal{H}_c$ be the state Hilbert space of \textquotedblleft coin\textquotedblright\ $c$ and $\mathcal{F}(\mathcal{H}_c)$ the free Fock space over $\mathcal{H}_c$. We write $$\mathcal{G}(\mathcal{H}_C)=\bigotimes_{c\in C}\mathcal{F}(\mathcal{H}_c).$$ We also assume that $\mathcal{H}$ is the state Hilbert space of the principal system. Let $\omega$ be the set of nonnegative integers. Then $\omega^C$ is the set of $C-$indexed tuples of nonnegative integers: $\overline{n}=\{n_c\}_{c\in C}$ with $n_c\in\omega$ for all $c\in C$, and we have:
$$\mathcal{G}(\mathcal{H}_C)\otimes\mathcal{H}\equiv \bigoplus_{\overline{n}\in\omega^C}\left(\bigotimes_{c\in C}\mathcal{H}_c^{\otimes n_c}\otimes\mathcal{H}\right).$$ 
We write $\mathcal{O}(\mathcal{G}(\mathcal{H}_C)\otimes \mathcal{H})$ for the set of all operators of the form $$\mathbf{A}=\sum_{\overline{n}\in\omega^C} \mathbf{A}(\overline{n}),$$ where $\mathbf{A}(\overline{n})$ is an operator in $\bigotimes_{c\in C}\mathcal{H}_c^{\otimes n_c}\otimes\mathcal{H}$ for each $\overline{n}\in\omega^C$. 
Recall that a binary relation $\sqsubseteq$ on a nonempty set $S$ if it is reflexive, transitive and antisymmetric. In this case, we often call $(S,\sqsubseteq)$ or even simply $S$ a partial order. We define a partial order $\leq$ on $\omega^C$ as follows: $\overline{n}\leq\overline{m}$ if and only if $n_c\leq m_c$ for all $c\in C$. A subset $\Omega\subseteq\omega^C$ is said to be below-closed if $\overline{n}\in\Omega$ and $\overline{m}\leq\overline{n}$ imply $\overline{m}\in\Omega$. We define the \textit{flat order} $\sqsubseteq$ on $\mathcal{O}(\mathcal{G}(\mathcal{H}_C)\otimes\mathcal{H})$ as follows: for any $\mathbf{A}=\sum_{\overline{n}\in\omega^C}^\infty\mathbf{A}(\overline{n})$ and $\mathbf{B}=\sum_{\overline{n}\in\omega^C}^\infty\mathbf{B}(\overline{n})$ in $\mathcal{O}(\mathcal{G}(\mathcal{H}_C)\otimes\mathcal{H})$, \begin{itemize}\item $\mathbf{A}\sqsubseteq\mathbf{B}$ if and only if there exists a below-closed subset $\Omega\subseteq\omega^C$ such that $\mathbf{A}(\overline{n})=\mathbf{B}(\overline{n})$ for all $\overline{n}\in\Omega$ and $\mathbf{A}(\overline{n})=0$ for all $\overline{n}\in\omega^C\setminus\Omega$.\end{itemize}
Let $(S,\sqsubseteq)$ be a partial order. A nonempty subset $T\subseteq S$ is called a chain if for any $x,y\in T$, it holds that $x\sqsubseteq y$ or $y\sqsubseteq x$. A partial order is said to be complete if it has the least element and every chain $T$ in it has the least upper bound $\bigsqcup T$. 
 
\begin{lem}\label{cpo-d} $(\mathcal{O}(\mathcal{G}(\mathcal{H}_C)\otimes\mathcal{H}),\sqsubseteq)$ is a complete partial order (CPO). 
\end{lem}

\begin{proof} First, $\sqsubseteq$ is reflexive because $\omega^C$ itself is below-closed. To show that $\sqsubseteq$ is transitive, we assume that $\mathbf{A}\sqsubseteq\mathbf{B}$ and $\mathbf{B}\sqsubseteq\mathbf{C}$. Then there exist below-closed $\Omega,\Gamma\subseteq\omega^C$ such that \begin{enumerate}
\item $\mathbf{A}(\overline{n})=\mathbf{B}(\overline{n})$ for all $\overline{n}\in\Omega$ and $\mathbf{A}(\overline{n})=0$ for all $\overline{n}\in\omega^C\setminus\Omega$;
\item $\mathbf{B}(\overline{n})=\mathbf{C}(\overline{n})$ for all $\overline{n}\in\Gamma$ and $\mathbf{B}(\overline{n})=0$ for all $\overline{n}\in\omega^C\setminus\Gamma$.
\end{enumerate} Clearly, $\Omega\cap\Gamma$ is below-closed, and $\mathbf{A}(\overline{n})=\mathbf{B}(\overline{n})=\mathbf{C}(\overline{n})$ for all $\overline{n}\in\Omega\cap\Gamma$. On the other hand, if $\overline{n}\in\omega^C\setminus(\Omega\cap\Gamma)=(\omega^C\setminus\Omega)\cup[\Omega\cap(\omega^C\setminus\Gamma)]$, then either $\overline{n}\in\omega^C\setminus\Omega$ and it follows from clause 1 that $\mathbf{A}(\overline{n})=0$, or $\overline{n}\in\Omega\cap(\omega^C\setminus\Gamma)$ and by combining clauses 1 and 2 we obtain $\mathbf{A}(\overline{n})=\mathbf{B}(\overline{n})=0$. Therefore, $\mathbf{A}\sqsubseteq\mathbf{C}$. Similarly, we can prove that $\sqsubseteq$ is antisymmetric. So, $(\mathcal{O}(\mathcal{G}(\mathcal{H}_C)\otimes \mathcal{H}),\sqsubseteq)$ is a partial order.

Obviously, the operator $\mathbf{A}=\sum_{\overline{n}\in\omega^C}\mathbf{A}(\overline{n})$ with $\mathbf{A}(\overline{n})=0$ (the zero operator in $\bigotimes_{c\in C}\mathcal{H}_c^{\otimes n_c}\otimes\mathcal{H}$) for all $\overline{n}\in\omega^C$ is the least element of $(\mathcal{O}(\mathcal{G}(\mathcal{H}_C)\otimes\mathcal{H}),\sqsubseteq)$. 
Now it suffices to show that any chain $\{\mathbf{A}_i\}$ in $(\mathcal{O}(\mathcal{G}(\mathcal{H}_C)\otimes \mathcal{H}),\sqsubseteq)$ has the least upper bound. For each $i$, we put \begin{align*}\Delta_i &=\{\overline{n}\in\omega^C:\mathbf{A}_i(\overline{n})=0\},\\ \Delta_i\downarrow &=\{\overline{m}\in\omega^C:\overline{m}\leq\overline{n}\ {\rm for\ some}\ \overline{n}\in\Delta_i\}.\end{align*}
Here $\Delta_i\downarrow$ is the below-completion of $\Delta_i$. Furthermore, we define operator $\mathbf{A}=\sum_{\overline{m}\in\omega^C}\mathbf{A}(\overline{n})$ as follows:\begin{equation*}
\mathbf{A}(\overline{n})=\begin{cases}\mathbf{A}_i(\overline{n})\ &{\rm if}\ \overline{n}\in\Delta_i\downarrow\ {\rm for\ some}\ i,\\
0 &{\rm if}\ \overline{n}\notin\bigcup_i(\Delta_i\downarrow).\end{cases}
\end{equation*}

\textit{Claim} 1: $\mathbf{A}$ is well-defined; that is, if $\overline{n}\in\Delta_i\downarrow$ and $\overline{n}\in\Delta_j\downarrow$, then $\mathbf{A}_i(\overline{n})=\mathbf{A}_j(\overline{n})$. In fact, since $\{\mathbf{A}_i\}$ is a chain, we have $\mathbf{A}_i\sqsubseteq\mathbf{A}_j$ or $\mathbf{A}_j\sqsubseteq\mathbf{A}_i.$ We only consider the case of $\mathbf{A}_i\sqsubseteq\mathbf{A}_j$ (the case of $\mathbf{A}_j\sqsubseteq\mathbf{A}_i$ is proved by duality). Then there exists below-closed $\Omega\subseteq\omega^C$ such that $\mathbf{A}_i(\overline{n})=\mathbf{A}_j(\overline{n})$ for all $\overline{n}\in\Omega$ and $\mathbf{A}(\overline{n})=0$ for all $\overline{n}\in\omega^C\setminus\Omega$. It follows from $\overline{n}\in\Delta_i\downarrow$ that $\overline{n}\sqsubseteq\overline{m}$ for some $\overline{m}$ with $\mathbf{A}_i(\overline{m})\neq 0$. Since $\overline{m}\notin\omega^C\setminus\Omega$, i.e. $\overline{m}\in\Omega$, we have $\overline{n}\in\Omega$ because $\Omega$ is below-closed. So, $\mathbf{A}_i(\overline{n})=\mathbf{A}_j(\overline{n})$. 

\textit{Claim} 2: $\mathbf{A}=\bigsqcup_i\mathbf{A}_i$. In fact, for each $i$, $\Delta_i\downarrow$ is below-closed, and $\mathbf{A}_i(\overline{n})=\mathbf{A}(\overline{n})$ for all $\overline{n}\in\Delta_i\downarrow$ and $\mathbf{A}_i(\overline{n})=0$ for all $\overline{n}\in\omega^C\setminus(\Delta_i\downarrow)$. So, $\mathbf{A}_i\sqsubseteq\mathbf{A}$, and $\mathbf{A}$ is an upper bound of $\{\mathbf{A}_i\}$. Now assume that $\mathbf{B}$ is an upper bound of $\{\mathbf{A}_i\}$: for all $i$, $\mathbf{A}_i\sqsubseteq\mathbf{B}$; that is, there exists below-closed $\Omega_i\subseteq\omega^C$ such that $\mathbf{A}_i(\overline{n})=\mathbf{B}(\overline{n})$ for all $\overline{n}\in\Omega_i$ and $\mathbf{A}_i(\overline{n})=0$ for all $\overline{n}\in\omega^C\setminus \Omega_i.$ By the definition of $\Delta_i$ and below-closeness of $\Omega_i$, we know that $\Delta_i\downarrow\subseteq\Omega_i$. We take $\Omega=\bigcup_i\left(\Delta_i\downarrow\right)$. Clearly, $\Omega$ is below-closed, and if  
$\overline{n}\in\omega^C\setminus\Omega$, then $\mathbf{A}(\overline{n})=0$. On the other hand, if $\overline{n}\in\Omega$, then for some $i$, we have $\overline{n}\in\Delta_i\downarrow$, and it follows that $\overline{n}\in\Omega_i$ and $\mathbf{A}(\overline{n})=\mathbf{A}_i(\overline{n})=\mathbf{B}(\overline{n})$. Therefore, $\mathbf{A}\sqsubseteq\mathbf{B}$. $\blacksquare$ \end{proof}

For any operators $\mathbf{A}=\sum_{\overline{n}\in\omega^C}\mathbf{A}(\overline{n})$ and $\mathbf{B}=\sum_{\overline{n}\in\omega^C}\mathbf{B}(\overline{n})$ in $\mathcal{O}(\mathcal{G}(\mathcal{H}_C)\otimes\mathcal{H})$, their product is naturally defined as \begin{equation}\label{ffp}\mathbf{A}\cdot \mathbf{B}=\sum_{\overline{n}\in\omega^C}\left(\mathbf{A}(\overline{n})\cdot\mathbf{B}(\overline{n})\right),\end{equation} which is also in $\mathcal{O}(\mathcal{G}(\mathcal{H}_C)\otimes\mathcal{H})$. We can define guarded composition of operators in Fock spaces by extending equation (\ref{qccq}). Let $c\in C$ and $\{|i\rangle\}$ be an orthonormal basis of $\mathcal{H}_c$, and let $\mathbf{A}_i=\sum_{\overline{n}\in\omega^C}\mathbf{A}_i(\overline{n})$ be an operator in $\mathcal{O}(\mathcal{G}(\mathcal{H}_C)\otimes\mathcal{H})$ for each $i$. Then the guarded composition of $\mathbf{A}_i$'s along with the basis $\{|i\rangle\}$ is \begin{equation}\label{ffgu}\square \left(c, |i\rangle\rightarrow\mathbf{A}_i\right)=\sum_{\overline{n}\in\omega^C}\left(\sum_i\left(|i\rangle_c\langle i|\otimes\mathbf{A}_i(\overline{n})\right)\right).\end{equation} Note that for each $\overline{n}\in\omega^C$, $\sum_i\left(|i\rangle_c\langle i|\otimes\mathbf{A}_i(n)\right)$ is an operator in $$\mathcal{H}_c^{\otimes (n_c+1)}\otimes\bigotimes_{d\in C\setminus\{c\}}\mathcal{H}_d^{n_d}\otimes \mathcal{H},$$ and thus $\square \left(c,|i\rangle\rightarrow\mathbf{A}_i\right)\in\mathcal{O}(\mathcal{G}(\mathcal{H}_C)\otimes\mathcal{H})$.
Recall that a mapping $f$ from a CPO $(S_1,\sqsubseteq)$ into another CPO $(S_2,\sqsubseteq)$ is said to be continuous if for any chain $T$ in $S_1$, its image $f(T)=\{f(x):x\in T\}$ under $f$ has the least upper bound and $\bigsqcup f(T)=f(\bigsqcup T)$. The following lemma shows that both product and guarded composition of operators in the free Fock space are continuous. 

\begin{lem}\label{contin-lem} Let $\{\mathbf{A}_j\}$, $\{\mathbf{B}_j\}$ and $\{\mathbf{A}_{ij}\}$ for each $i$ be chains in $(\mathcal{O}(\mathcal{G}(\mathcal{H}_C)\otimes\mathcal{H}),\sqsubseteq)$. Then \begin{enumerate}\item $\bigsqcup_j\left(\mathbf{A}_j\cdot\mathbf{B}_j\right)=\left(\bigsqcup_j\mathbf{A}_j\right)\cdot\left(\bigsqcup_j\mathbf{B}_j\right).$ \item $\bigsqcup_j \square\left(c, |i\rangle\rightarrow\mathbf{A}_{ij}\right)=\square\left(c,|i\rangle\rightarrow\left(\bigsqcup_j\mathbf{A}_{ij}\right)\right).$
\end{enumerate}\end{lem}

\begin{proof} We only prove part 2. The proof of part 1 is similar. For each $i$, we assume that $$\bigsqcup_j\mathbf{A}_{ij}=\mathbf{A}_i=\sum_{\overline{n}\in\omega^C}\mathbf{A}_i(\overline{n}).$$ By the construction of least upper bound in $(\mathcal{O}(\mathcal{G}(\mathcal{H}_C)\otimes\mathcal{H}),\sqsubseteq)$ given in the proof of Lemma \ref{cpo-d}, we can write $\mathbf{A}_{ij}=\sum_{\overline{n}\in\Omega_{ij}}\mathbf{A}_i(\overline{n})$ for some $\Omega_{ij}\subseteq\omega^C$ with $\bigcup_j \Omega_{ij}=\omega^C$ for every $i$. By appending zero operators to the end of shorter summations, we may further ensure that $\Omega_{ij}$'s for all $i$ are the same, say $\Omega_j$. Then by the defining equation (\ref{ffgu}) we obtain: \begin{equation*}\begin{split}\bigsqcup_j\square\left(c,|i\rangle\rightarrow\mathbf{A}_{ij}\right)&=\bigsqcup_j\sum_{\overline{n}\in\Omega_j}\left(\sum_i\left(|i\rangle_c\langle i|\otimes\mathbf{A}_i(\overline{n})\right)\right)\\ &=\sum_{\overline{n}\in\omega^C}\left(\sum_i
\left(|i\rangle_c\langle i|\otimes\mathbf{A}_i(\overline{n})\right)\right)=\square\left(c, |i\rangle\rightarrow\mathbf{A}_i\right).\ \blacksquare\end{split}\end{equation*}\end{proof}

\subsection{Semantic Functionals of Program Schemes}

Let $P=P[X_1,...,X_m]$ be a program scheme. We write $C$ for the set of \textquotedblleft coins\textquotedblright\ occuring in $P$. For each $c\in C$, let $\mathcal{H}_c$ be the state Hilbert space of quantum \textquotedblleft coin\textquotedblright\ $c$. As said in Section \ref{sec-syn}, the principal system of $P$ is the composition of the systems denoted by principal variables appearing in $P$. Let $\mathcal{H}$ be the state Hilbert space of the principal system. 

\begin{defn}\label{SDF} The semantic functional of program scheme $P$ is a mapping $$\llbracket P\rrbracket: \mathcal{O}(\mathcal{G}(\mathcal{H}_C)\otimes \mathcal{H})^m\rightarrow  \mathcal{O}(\mathcal{G}(\mathcal{H}_C)\otimes\mathcal{H}).$$ For any operators $\mathbf{A}_1,...,\mathbf{A}_m\in\mathcal{O}(\mathcal{G}(\mathcal{H}_C)\otimes\mathcal{H})$, $\llbracket P\rrbracket (\mathbf{A}_1,...,\mathcal{A}_m)$ is inductively defined as follows: \begin{enumerate}\item If $P=\mathbf{abort},$ then  $\llbracket P\rrbracket (\mathbf{A}_1,...,\mathbf{A}_m)$ is the zero operator in $\mathbf{A}=\sum_{\overline{n}\in\omega^C}\mathbf{A}(\overline{n})$ with $\mathbf{A}(\overline{n})=0$ (the zero operator in $\bigotimes_{c\in C}\mathcal{H}_c^{\otimes n_c}\otimes\mathcal{H}$) for all $\overline{n}\in\omega^C$;
\item If $P= \mathbf{skip}$, then $\llbracket P\rrbracket (\mathbf{A}_1,...,\mathbf{A}_m)$ is the identity operator $\mathbf{A}=\sum_{\overline{n}\in\omega^C}\mathbf{A}(\overline{n})$ with $\mathbf{A}(\overline{n})=I$ (the identity  operator in $\bigotimes_{c\in C}\mathcal{H}_c^{\otimes n_c}\otimes\mathcal{H}$) for all $\overline{n}\in\omega^C$ with $n_c\neq 0$ for every $c\in C$; 
\item If $P=U[\overline{c}, \overline{q}]$, then $\llbracket P\rrbracket (\mathbf{A}_1,...,\mathbf{A}_m)$ is the cylindrical extension of $U$: $\mathbf{A}=\sum_{\overline{n}\in\omega^C}\mathbf{A}(\overline{n})$ with $\mathbf{A}(\overline{n})=I_1\otimes I_2(\overline{n})\otimes U\otimes I_3$, where:\begin{enumerate}\item $I_1$ is the identity operator in the state Hilbert space of those \textquotedblleft coins\textquotedblright\ that are not in $\overline{c}$;\item $I_2(\overline{n})$ is the identity operator in $\bigotimes_{c\in\overline{c}}\mathcal{H}_c^{\otimes (n_c-1)}$; and \item $I_3$ is the identity operator in the state Hilbert space of those principal variables that are not in $\overline{q}$ for all $n\geq 1$ ;\end{enumerate}
\item If $P=X_j$ $(1\leq j\leq m)$, then $\llbracket P\rrbracket (\mathbf{A}_1,...,\mathbf{A}_m)=\mathbf{A}_j$; \item If $P=P_1;P_2$, then $$\llbracket P\rrbracket (\mathbf{A}_1,...,\mathbf{A}_m)=\llbracket P_2\rrbracket (\mathbf{A}_1,...,\mathbf{A}_m)\cdot\llbracket P_1\rrbracket (\mathbf{A}_1,...,\mathbf{A}_m)$$ (see the defining equation (\ref{ffp}) of product of operators in the free Fock space); \item If $P=\mathbf{qif}\ [c](\square i\cdot |i\rangle\rightarrow P_i)\ \mathbf{fiq}$, then
$$\llbracket P\rrbracket (\mathbf{A}_1,...,\mathbf{A}_m)=\square \left(c,|i\rangle\rightarrow\llbracket P_i\rrbracket (\mathbf{A}_1,...,\mathbf{A}_m)\right)$$
(see the defining equation (\ref{ffgu}) of guarded composition of operators in the free Fock space).
\end{enumerate}
\end{defn}

Whenever $m=0$; that is, $P$ contains no procedure identifiers, then the above definition degenerates to Definition \ref{seman-w}.

The cartesian power $\mathcal{O}(\mathcal{G}(\mathcal{H}_C)\otimes \mathcal{H})^m$ is naturally equipped with the order $\sqsubseteq$ defined componently from the order in $\mathcal{O}(\mathcal{G}(\mathcal{H}_C)\otimes \mathcal{H})$: for any $\mathbf{A}_1,...,\mathbf{A}_m,\mathbf{B}_1,...,\mathbf{B}_m\in\mathcal{O}(\mathcal{G}(\mathcal{H}_C)\otimes \mathcal{H})$, \begin{itemize}\item $(\mathbf{A}_1,...,\mathbf{A}_m)\sqsubseteq (\mathbf{B}_1,...,\mathbf{B}_m)$ if and only if for every $1\leq i\leq m$, $\mathbf{A}_i\sqsubseteq \mathbf{B}_i.$\end{itemize} Then $(\mathcal{O}(\mathcal{G}(\mathcal{H}_C)\otimes \mathcal{H}_q)^m,\sqsubseteq)$ is a CPO too. Furthermore, we have:

\begin{thm}\label{continuity-1} (\textbf{Continuity of Semantic Functionals}) The semantic functional $\llbracket P\rrbracket: (\mathcal{O}(\mathcal{G}(\mathcal{H}_C)\otimes \mathcal{H})^m,\sqsubseteq)\rightarrow (\mathcal{O}(\mathcal{G}(\mathcal{H}_C)\otimes \mathcal{H}),\sqsubseteq)$ is continuous.\end{thm}
\begin{proof} It can be easily proved by induction on the structure of $P$ using Lemma \ref{contin-lem}. $\blacksquare$\end{proof}

For each \textquotedblleft coin\textquotedblright\ $c\in C$, we introduce the creation functional: $\mathbb{K}_c:\mathcal{O}(\mathcal{G}(\mathcal{H}_C)\otimes \mathcal{H})\rightarrow \mathcal{O}(\mathcal{G}(\mathcal{H}_C)\otimes \mathcal{H})$ defined as follows: for any $\mathbf{A}=\sum_{\overline{n}\in\omega^C}\mathbf{A}(\overline{n})\in \mathcal{O}(\mathcal{G}(\mathcal{H}_C)\otimes \mathcal{H})$, $$\mathbb{K}_c(\mathbf{A})=\sum_{\overline{n}\in\omega^C}(I_c\otimes\mathbf{A}(\overline{n}))$$ where $I_c$ is the identity operator in $\mathcal{H}_c$. We observe that $\mathbf{A}(\overline{n})$ is an operator in $\bigotimes_{d\in C}$ $\mathcal{H}_d^{\otimes n_d}\otimes\mathcal{H}$, whereas $I_c\otimes\mathbf{A}(\overline{n})$ is an operator in $\mathcal{H}_c^{\otimes (n_c+1)}\otimes\bigotimes_{d\in C\setminus\{d\}}\mathcal{H}_d^{\otimes n_d}\otimes\mathcal{H}$. 
Intuitively, the creation functional $\mathbb{K}_c$ moves all copies of $\mathcal{H}_c$ one position to the right so that $i$th copy becomes $(i+1)$th copy for all $i=0,1,2,....$ Thus, a new position is created at the left end for a new copy of $\mathcal{H}_c$. For other \textquotedblleft coins\textquotedblright\ $d$, $\mathbb{K}_c$ does not move any copy of $\mathcal{H}_d$. It is clear that for any two \textquotedblleft coins\textquotedblright\ $c,d$, $\mathbb{K}_c$ and $\mathbb{K}_d$ commute; that is, $\mathbb{K}_a\circ\mathbb{K}_d=\mathbb{K}_d\circ\mathbb{K}_c$. Note that the set $C$ of \textquotedblleft coins\textquotedblright\ in $P$ is finite. Suppose that $C=\{c_1,c_2,...,c_k\}$. Then we can define the creation functional $$\mathbb{K}_C=\mathbb{K}_{c_1}\circ\mathbb{K}_{c_2}\circ ...\circ\mathbb{K}_{c_k}.$$ For the special case where the set $C$ of \textquotedblleft coins\textquotedblright\ is empty, $\mathbb{C}$ is the identity functional; that is, $\mathbb{C}(\mathbf{A})=\mathbf{A}$ for all $\mathbf{A}$. 

\begin{lem}\label{continuity-2} For each $c\in C$, the creation functionals $\mathbb{K}_c$ and $\mathbb{K}_C:(\mathcal{O}(\mathcal{G}(\mathcal{H}_C)\otimes \mathcal{H}),$ $\sqsubseteq)\rightarrow (\mathcal{O}(\mathcal{G}(\mathcal{H}_C)\otimes \mathcal{H}),\sqsubseteq)$ are continuous.\end{lem}

\begin{proof} Straightforward by definition. $\blacksquare$\end{proof}

Combining continuity of semantic functional and the creation functional (Theorem \ref{continuity-1} and Lemma \ref{continuity-2}), we obtain: 

\begin{cor}\label{corr1} Let $P=P[X_1,...,X_m]$ be a program scheme and $C$ the set of \textquotedblleft coins\textquotedblright\ occurring in $P$. Then the functional: \begin{equation*}\begin{split}& \mathbb{K}_C^{m}\circ\llbracket P\rrbracket: (\mathcal{O}(\mathcal{G}(\mathcal{H}_C)\otimes \mathcal{H})^{m},\sqsubseteq)\rightarrow (\mathcal{O}(\mathcal{G}(\mathcal{H}_C)\otimes \mathcal{H}),\sqsubseteq),\\  &(\mathbb{K}_C^{m}\circ \llbracket P\rrbracket)(\mathbf{A}_1,...,\mathbf{A}_m)=\llbracket P\rrbracket(\mathbb{K}_C(\mathbf{A}_1),...,\mathbb{K}_C(\mathbf{A}_m))\end{split}\end{equation*} for any $\mathbf{A}_1,...,\mathbf{A}_m\in \mathcal{O}(\mathcal{G}(\mathcal{H}_C)\otimes \mathcal{H})$, is continuous.\end{cor}
 
\subsection{Fixed Point Semantics}
 
Now we are ready to define the denotational semantics of recursive programs using the standard fixed point technique. Let us consider a recursive program $P$ declared by the system of recursive equations:\begin{equation}\label{Decla}D: \begin{cases}X_1\Leftarrow P_1,\\ \ \ \ \ \ \ ......\\ X_m\Leftarrow P_m,\end{cases}\end{equation} where $P_i=P_i[X_1,...,X_m]$ is a program scheme containing at most procedure identifiers $X_1,...,X_m$ for every $1\leq i\leq m$. The system $D$ of recursive equations naturally induces a semantic functional: \begin{equation}\label{seman-rec}\begin{split}&\llbracket D\rrbracket:\mathcal{O}(\mathcal{G}(\mathcal{H}_C)\otimes \mathcal{H})^{m}\rightarrow \mathcal{O}(\mathcal{G}(\mathcal{H}_C)\otimes \mathcal{H})^m,\\ &\llbracket D\rrbracket (\mathbf{A}_1,...,\mathbf{A}_m)=((\mathbb{K}_C^{m}\circ\llbracket P_1\rrbracket)(\mathbf{A}_1,...,\mathbf{A}_m),...,\\ &\ \ \ \ \ \ \ \ \ \ \ \ \ \ \ \ \ \ \ \ \ \ \ \ \ \ \ \ \ \ \ \ \ \ \ \ \ \ \ \ \ \ \ \ \ \ \ \ \ \ \ \ \ \ \ \ \ \ \ \ \ \ \ \ \ \ \ \ \ \ \ \ \ \ (\mathbb{K}_C^{m}\circ\llbracket P_m\rrbracket)(\mathbf{A}_1,...,\mathbf{A}_m))\end{split}\end{equation} for all $\mathbf{A}_1,...,\mathbf{A}_m\in \mathcal{O}(\mathcal{G}(\mathcal{H}_C)\otimes \mathcal{H})$, where $C$ is the set of \textquotedblleft coins\textquotedblright\ appearing in $D$; that is, in one of $P_1,...,P_m$. It follows from Theorem 4.20 in \cite{LS87} and Corollary \ref{corr1} that 
$\llbracket D\rrbracket: (\mathcal{O}(\mathcal{G}(\mathcal{H}_C)\otimes \mathcal{H})^{m},\sqsubseteq)\rightarrow (\mathcal{O}(\mathcal{G}(\mathcal{H}_C)\otimes \mathcal{H})^m,\sqsubseteq)$ is continuous. Then the Knaster-Tarski Fixed Point Theorem asserts that $\llbracket D\rrbracket$ has the least fixed point $\mu\llbracket D\rrbracket.$

\begin{defn}\label{def-fis}The fixed point (denotational) semantics of the recursive program $P$ declared by $D$ is $$\llbracket P\rrbracket_{fix} =\llbracket P\rrbracket (\mu\llbracket D\rrbracket);$$ that is, if $\mu\llbracket D\rrbracket =(\mathbf{A}_1^\ast,...,\mathbf{A}_m^\ast)\in \mathcal{O}(\mathcal{G}(\mathcal{H}_C)\otimes\mathcal{H})^m$, then $\llbracket P\rrbracket_{fix} =\llbracket P\rrbracket (\mathbf{A}_1^\ast,...,\mathbf{A}_m^\ast)$ (see Definition \ref{SDF}).\end{defn}

\subsection{Syntactic Approximation} 
We now turn to consider the syntactic approximation technique for defining the semantics of recursive programs. As discussed at the end of Section \ref{Rewk} and further clarified in Example \ref{ex4}, a problem that was not present in the classical case is that we have to carefully avoid the conflict of quantum \textquotedblleft coin\textquotedblright\ variables when defining the notion of substitution. To overcome it, we assume that each \textquotedblleft coin\textquotedblright\ variable $c\in C$ has infinitely many copies $c_0, c_1,c_2,...$ with $c_0=c$. The variables $c_1,c_2,...$ are used to represent a sequence of particles that are all identical to the particle $c_0=c$. Then the notion of program scheme defined in Section \ref{sec-syn} will be used in a slightly broader way: a program scheme may contain not only a \textquotedblleft coin\textquotedblright\ $c$ but also some of its copies $c_1,c_2,...$. If such a generalised program scheme contains no procedure identifiers, then it is called a generalised program. With these assumptions, we can introduce the notion of substitution. 

\begin{defn}\label{stut} Let $P=P[X_1,...,X_m]$ be a generalised program scheme that contains at most procedure identifiers $X_1,...,$ $X_m$, and let $Q_1,...,Q_m$ be generalised programs without any procedure identifier. Then the simultaneous substitution $P[Q_1/X_1,...,Q_m/$ $X_m]$ of $X_1,...,X_m$ by $Q_1,...,Q_m$ in $P$ is inductively defined as follows:
\begin{enumerate}\item If $P=\mathbf{abort}, \mathbf{skip}$ or an unitary transformation, then $P[Q_1/X_1,...,Q_m/X_m]=P$;
\item If $P=X_i$ $(1\leq i\leq m)$, then $P[Q_1/X_1,...,Q_m/X_m]=Q_i$; \item If $P=P_1;P_2$, then $$P[Q_1/X_1,...,Q_m/X_m]=P_1[Q_1/X_1,...,Q_m/X_m];P_2[Q_1/X_1,...,Q_m/X_m].$$ \item If $P=\mathbf{qif}\ [c](\square i\cdot |i\rangle\rightarrow P_i)\ \mathbf{fiq}$, then
$$P[Q_1/X_1,...,Q_m/X_m]=\mathbf{qif}\ [c](\square i\cdot |i\rangle\rightarrow P_i^\prime)\ \mathbf{fiq}$$ where for every $i$, $P_i^\prime$ is obtained through replacing the $j$th copy $c_j$ of $c$ in $P_i[Q_1/X_1,$ $...,Q_m/X_m]$ by the $(j+1)$th copy $c_{j+1}$ of $c$ for all $j$.\end{enumerate}\end{defn}

Note that in Clause 4 of the above definition, since $P$ is a generalised program scheme, the \textquotedblleft coin\textquotedblright\ $c$ may not be an original \textquotedblleft coin\textquotedblright\ but some copy $d_k$ of an original \textquotedblleft coin\textquotedblright\ $d\in C$. In this case, the $j$th copy of $c$ is actually the $(k+j)$th copy of $d$: $c_j=(d_k)_j=d_{k+j}$ for $j\geq -d.$ 

The semantics of a generalised program $P$ can be given using Definition \ref{seman-w} in the way where a \textquotedblleft coin\textquotedblright\ $c$ and its copies $c_1,c_2,...$ are treated as distinct variables to each other. For each \textquotedblleft coin\textquotedblright\ $c$, let $n_c$ be the greatest index $n$ such that the copy $c_n$ appears in $P$. Then the semantics $\llbracket P\rrbracket$ of $P$ is an operator in $\bigotimes_{c\in C}\mathcal{H}_c^{\otimes n_c}\otimes\mathcal{H}$. Furthermore, it can be identified with its cylindrical extension in $\mathcal{O}(\mathcal{G}(\mathcal{H}_C)\otimes \mathcal{H})$: $\sum_{\overline{m}\in\omega^C}\left(I(\overline{m})\otimes\llbracket P\rrbracket\right),$ where for each $\overline{m}\in\omega^C$, $I(\overline{m})$ is the identity operator in $\bigotimes_{c\in C}\mathcal{H}_c^{\otimes m_c}$. Based on this observation, the semantics of substitution defined above is characterised by the following:

\begin{lem}\label{lem-stut} For any (generalised) program scheme $P=P[X_1,...,X_m]$ and (generalised) programs $Q_1,...,Q_m$, we have: \begin{align*}\llbracket P[Q_1/X_1,...,Q_m/X_m]\rrbracket &= (\mathbb{K}_C^{m}\circ \llbracket P\rrbracket) (\llbracket Q_1\rrbracket,...,\llbracket Q_m\rrbracket)\\ &= \llbracket P\rrbracket (\mathbb{K}_C(\llbracket Q_1\rrbracket),...,\mathbb{K}_C(\llbracket Q_m\rrbracket)),\end{align*} where $\mathbb{K}_C$ is the creation functional with $C$ being the set of \textquotedblleft coins\textquotedblright\ in $P$.\end{lem}

\begin{proof} We prove the lemma by induction on the structure of $P$.

Case 1. $P=\mathbf{abort}$, $\mathbf{skip}$ or an unitary transformation. Obvious. 

Case 2. $P=X_j$ $(1\leq j\leq m)$. Then $P[Q_1/X_1,...,Q_m/X_m]=Q_m$. On the other hand, since the set of \textquotedblleft coins\textquotedblright\ in $P$ is empty, $\mathbb{K}_C(\llbracket Q_i\rrbracket)=\llbracket Q_i\rrbracket$ for all $1\leq i\leq m.$ Thus, by clause 4 of Definition \ref{SDF} we obtain: \begin{align*}\llbracket P[Q_1/X_1,&...,Q_m/X_m]\rrbracket = 
\llbracket Q_m\rrbracket\\ &=\llbracket P\rrbracket (\llbracket Q_1\rrbracket,...,\llbracket Q_m\rrbracket)=
\llbracket P\rrbracket (\mathbb{K}_C(\llbracket Q_1\rrbracket),...,\mathbb{K}_C(\llbracket Q_m\rrbracket)).\end{align*}
 
 Case 3. $P=P_1;P_2$. Then by clause 3 of Definition \ref{seman-w}, clause 5 of Definition \ref{SDF} and the induction hypothesis, we have: \begin{align*}
 \llbracket P[Q_1/X_1,&...,Q_m/X_m]\rrbracket = \llbracket P_1[Q_1/X_1,...,Q_m/X_m];P_2[Q_1/X_1,...,Q_m/X_m]\rrbracket\\ 
 &=\llbracket P_2[Q_1/X_1,...,Q_m/X_m]\rrbracket\cdot \llbracket P_1[Q_1/X_1,...,Q_m/X_m]\rrbracket\\ 
&= \llbracket P_2\rrbracket (\mathbb{K}_C(\llbracket Q_1\rrbracket),...,\mathbb{K}_C(\llbracket Q_m\rrbracket))\cdot \llbracket P_1\rrbracket (\mathbb{K}_C(\llbracket Q_1\rrbracket),...,\mathbb{K}_C(\llbracket Q_m\rrbracket))\\ &=\llbracket P_1;P_2\rrbracket (\mathbb{K}_C(\llbracket Q_1\rrbracket),...,\mathbb{K}_C(\llbracket Q_m\rrbracket))\\ &=\llbracket P\rrbracket (\mathbb{K}_C(\llbracket Q_1\rrbracket),...,\mathbb{K}_C(\llbracket Q_m\rrbracket)). 
 \end{align*} 
 
 Case 4. $P=\mathbf{qif}\ [c](\square i\cdot |i\rangle\rightarrow P_i)\ \mathbf{fiq}$. Then $$P[Q_1/X_1,...,Q_m/X_m]=\mathbf{qif}\ [c](\square i\cdot |i\rangle\rightarrow P^\prime_i)\ \mathbf{fiq},$$ where $P_i^\prime$ is obtained according to clause 4 of Definition \ref{stut}. For each $i$, by the induction hypothesis we obtain: 
$$\llbracket P_i[Q_1/X_1,...,Q_m/X_m]\rrbracket = \llbracket P_i\rrbracket (\mathbb{K}_{C\setminus\{c\}}(\llbracket Q_1\rrbracket),...,\mathbb{K}_{C\setminus\{c\}}(\llbracket Q_m\rrbracket))$$
because the \textquotedblleft coin\textquotedblright\ $c$ does not appear in $P_i^\prime$. Furthermore, it follows that \begin{align*}
\llbracket P_i^\prime\rrbracket &=\mathbb{K}_c(\llbracket P_i[Q_1/X_1,...,Q_m/X_m]\rrbracket)\\ &= \mathbb{K}_c(\llbracket P_i\rrbracket (\mathbb{K}_{C\setminus\{c\}}(\llbracket Q_1\rrbracket),...,\mathbb{K}_{C\setminus\{c\}}(\llbracket Q_m\rrbracket)))\\ &=\llbracket P_i\rrbracket ((\mathbb{K}_c\circ \mathbb{K}_{C\setminus\{c\}})(\llbracket Q_1\rrbracket),...,(\mathbb{K}_c\circ \mathbb{K}_{C\setminus\{c\}})(\llbracket Q_m\rrbracket))\\ &=\llbracket P_i\rrbracket (\mathbb{K}_C(\llbracket Q_1\rrbracket),...,\mathbb{K}_C(\llbracket Q_m\rrbracket)). 
\end{align*} Therefore, by clause 4 of Definition \ref{seman-w}, clause 6 of Definition \ref{SDF} and equation (\ref{ffgu}), we have: 
\begin{align*}\llbracket P[Q_1/X_1,...,Q_m/X_m]\rrbracket &= \sum_i \left(|i\rangle\langle i|\otimes \llbracket P^\prime_i\rrbracket\right)\\ &=\square (c,|i\rangle\rightarrow 
\llbracket P_i\rrbracket (\mathbb{K}_C(\llbracket Q_1\rrbracket),...,\mathbb{K}_C(\llbracket Q_m\rrbracket))\\ 
&=\llbracket P\rrbracket (\mathbb{K}_C(\llbracket Q_1\rrbracket),...,\mathbb{K}_C(\llbracket Q_m\rrbracket)).\ 
\blacksquare\end{align*}\end{proof}

The notion of syntactic approximation can be defined based on Definition \ref{stut}. 

\begin{defn}\label{synap}\begin{enumerate}\item Let $X_1,...,X_m$ be procedure identifiers declared by the system $D$ of recursive equations (\ref{Decla}). Then for each $1\leq k\leq m$, the $n$th syntactic approximation $X_k^{(n)}$ of $X_k$ is inductively defined as follows:\begin{equation*}\begin{cases}& X_k^{(0)}=\mathbf{abort},\\
& X_k^{(n+1)}=P_k[X_1^{(n)}/X_1,...,X_m^{(n)}/X_m]\ {\rm for}\ n\geq 0.\end{cases}\end{equation*}
\item Let $P=P[X_1,...,X_m]$ be a recursive program declared by the system $D$ of equations (\ref{Decla}). Then for each $n\geq 0$, its $n$th syntactic approximation $P^{(n)}$ is inductively defined as follows: \begin{equation*}\begin{cases}& P^{(0)}=\mathbf{abort},\\
& P^{(n+1)}=P[X_1^{(n)}/X_1,...,X_m^{(n)}/X_m]\ {\rm for}\ n\geq 0.\end{cases}\end{equation*}\end{enumerate}
\end{defn}

Syntactic approximation actually gives an operational semantics of quantum recursive programs. As in the theory of classical programming, substitution represents an application of the so-called \textit{copy rule} - at runtime a procedure call is treated like the procedure body inserted at the place of call (see, for example, \cite{Hoa71}). Of course, simplification may happen within $X_k^{(n)}$ by operations of linear operators; for example, $C[q_1,q_2];X[q_2];C[q_1,q_2]$ can be replaced by $X[q_2]$, where $q_1,q_2$ are principal system variables, $C$ is the CNOT gate and $X$ is the NOT gate. To simplify the presentation, we choose not to explicitly describe simplification.    

The major difference between the classical case and the quantum case is that in the latter we need to continuously introduce new \textquotedblleft coin\textquotedblright\ variables to avoid variable conflict when we unfold a quantum recursive program using its syntactic approximations: for each $n\geq 0$, a new copy of each \textquotedblleft coin\textquotedblright\ in $P_k$ is created in the substitution $X_k^{(n+1)}=P[X_1^{(n)}/X_1,...,X_m^{(n)}/X_m]$ (see Clause 4 of Definition \ref{stut}). Thus, a quantum recursive program should be understood as a quantum system with variable particle number and described in the second quantisation formalism.
Note that for all $1\leq k\leq m$ and $n\geq 0$, the syntactic approximation $X_k^{(n)}$ is a generalised program containing no procedure identifiers. Thus, its semantics $\llbracket X_k^{(n)}\rrbracket$ can be given by a slightly extended version of Definition \ref{seman-w}: a \textquotedblleft coin\textquotedblright $c$ and its copies $c_1,c_2,...$ are allowed to appear in the same (generalised) program and they are considered as distinct variables. As before, the principal system is the composite system of the subsystems denoted by principal variables appearing in $P_1,...,P_m$ and its state Hilbert space is denoted by $\mathcal{H}$. Assume that $C$ is the set of \textquotedblleft coin\textquotedblright\ variables appearing in $P_1,...,P_m$. For each $c\in C$, we write $\mathcal{H}_c$ for the state Hilbert space of quantum \textquotedblleft coin\textquotedblright\ $c$.  Then it is easy to see that $\llbracket X_k^{(n)}\rrbracket$ is an operator in $\bigoplus_{j=0}^n\left(\mathcal{H}_C^{\otimes n}\otimes\mathcal{H}\right)$, where $\mathcal{H}_C=\bigotimes_{c\in C}\mathcal{H}_c$. So, we can imagine that $\llbracket X_k^{(n)}\rrbracket\in\mathcal{O}(\mathcal{G}(\mathcal{H}_C)\otimes\mathcal{H})$. Furthermore, we have:

\begin{lem} For each $1\leq k\leq m$, $\{\llbracket X_k^{(n)}\rrbracket\}_{n=0}^\infty$ is an increasing chain and thus \begin{equation}\label{infty-v}\llbracket X_k^{(\infty)}\rrbracket=\lim_{n\rightarrow\infty}\llbracket X_k^{(n)}\rrbracket\stackrel{\triangle}{=}\bigsqcup_{n=0}^\infty\llbracket X_k^{(n)}\rrbracket\end{equation} exists in $(\mathcal{O}(\mathcal{G}(\mathcal{H}_C)\otimes\mathcal{H}),\sqsubseteq)$.
\end{lem}

\begin{proof} We show that $\llbracket X_k^{(n)}\rrbracket\sqsubseteq \llbracket X_k^{(n+1)}\rrbracket$ by induction on $n$. The case of $n=0$ is trivial because $\llbracket X_k^{(0)}\rrbracket=\llbracket\mathbf{abort}\rrbracket =0$. In general, by the induction hypothesis on $n-1$ and Corollary \ref{corr1}, we have: \begin{align*}\llbracket X_k^{(n)}\rrbracket &=\llbracket P_k\rrbracket(\mathbb{K}_C(\llbracket X_1^{(n-1)}\rrbracket),...,\mathbb{K}_C(\llbracket X_m^{(n-1)}\rrbracket))\\ &\sqsubseteq \llbracket P_k\rrbracket
(\mathbb{K}_C(\llbracket X_1^{(n)}\rrbracket),...,\mathbb{K}_C(\llbracket X_m^{(n)}\rrbracket))=\llbracket X_k^{(n+1)}\rrbracket,
\end{align*} where $C$ is the set of \textquotedblleft coins\textquotedblright\ in $D$. Then existence of the least upper bound (\ref{infty-v}) follows immediately from Lemma \ref{cpo-d}. $\blacksquare$
\end{proof}

\begin{defn}\label{def-ops}Let $P$ be a recursive program declared by the system $D$ of equations (\ref{Decla}). Then its operational semantics is $$\llbracket P\rrbracket_{op}=\llbracket P\rrbracket(\llbracket X_1^{(\infty)}\rrbracket,...,\llbracket X_m^{(\infty)}\rrbracket).$$
\end{defn}

The operational semantics of recursive program $P$ can be characterised by the limit of its syntactic approximations (with respect to its declaration $D$).

\begin{prop}\label{prop-sapp}$\llbracket P\rrbracket_{op}=\bigsqcup_{n=0}^\infty\llbracket P^{(n)}\rrbracket.$\end{prop}
\begin{proof} It follows from Lemma \ref{lem-stut} that \begin{align*}
\bigsqcup_{n=0}^\infty\llbracket P^{(n)}\rrbracket &= \bigsqcup_{n=0}^\infty\llbracket P^{(n)}\rrbracket\\ 
&=\bigsqcup_{n=0}^\infty\llbracket P[X_1^{(n)}/X_1,...,X_m^{(n)}/X_m]\rrbracket\\ 
&=\bigsqcup_{n=0}^\infty\llbracket P\rrbracket (\mathbb{K}_C(\llbracket X_1^{(n)}\rrbracket),...,\mathbb{K}_C(\llbracket X_m^{(n)}\rrbracket))
\end{align*} where $\mathbb{K}_C$ is the creation functional with respect to the \textquotedblleft coins\textquotedblright\ $C$ in $P$. However, all the \textquotedblleft coins\textquotedblright\ $C$ in $P$ do not appear in $X_1^{(n)},...,X_m^{(n)}$ (see the condition in Definition \ref{RecP-DF}.2). So, $\mathbb{K}_C(\llbracket X_k^{(n)}\rrbracket)=\llbracket X_k^{(n)}\rrbracket$ for every $1\leq k\leq m$, and by Theorem \ref{continuity-1} we obtain:  
\begin{align*}
\bigsqcup_{n=0}^\infty\llbracket P^{(n)}\rrbracket &=\bigsqcup_{n=0}^\infty\llbracket P\rrbracket (\llbracket X_1^{(n)}\rrbracket,...,\llbracket X_m^{(n)}\rrbracket)\\
&=\llbracket P\rrbracket \left(\bigsqcup_{n=0}^\infty\llbracket X_1^{(n)}\rrbracket,...,\bigsqcup_{n=0}^\infty\llbracket X_m^{(n)}\rrbracket\right)\\
&=\llbracket P\rrbracket (\llbracket X_1^{\infty}\rrbracket,...,\llbracket X_m^{\infty}\rrbracket)=\llbracket P\rrbracket_{op}.\ \blacksquare
\end{align*} 
\end{proof}

The equivalence between denotational and operational semantics of recursive programs is established in the following:  
\begin{thm}\label{ss-equv} (\textbf{Equivalence of Denotational Semantics and Operational Semantics}) For any recursive program $P$, we have $\llbracket P\rrbracket_{fix}=\llbracket P\rrbracket_{op}.$
\end{thm}

\begin{proof} By Definitions \ref{def-fis} and \ref{def-ops}, it suffices to show that $(\llbracket X_1^{(\infty)}\rrbracket,...,\llbracket X_m^{(\infty)}\rrbracket)$ is the least fixed point of semantic functional $\llbracket D\rrbracket$, where $D$ is the declaration of procedure identifiers in $P$. With Theorem \ref{continuity-1} and Lemmas \ref{continuity-2} and \ref{lem-stut}, we obtain:\begin{equation*}\begin{split}\llbracket X_k^{(\infty)}\rrbracket =\bigsqcup_{n=0}^\infty \llbracket X_k^{(n)}\rrbracket
&=\bigsqcup_{n=0}^\infty \llbracket P_k[X_1^{(n)}/X_1,...,X_m^{(n)}/X_m]\rrbracket\\ 
&=\bigsqcup_{n=0}^\infty \llbracket P_k\rrbracket (\mathbb{K}_C(\llbracket X_1^{(n)}\rrbracket),...,\mathbb{K}_C(\llbracket X_m^{(n)}\rrbracket))\\ 
&= \llbracket P_k\rrbracket \left(\mathbb{K}_C\left(\bigsqcup_{n=0}^\infty\llbracket X_1^{(n)}\rrbracket\right),...,\mathbb{K}_C\left(\bigsqcup_{n=0}^\infty\llbracket X_m^{(n)}\rrbracket\right)\right)\\
&=\llbracket P_k\rrbracket (\mathbb{K}_C(\llbracket X_1^{(\infty)}\rrbracket),...,\mathbb{K}_C(\llbracket X_m^{(\infty)}\rrbracket))\end{split}\end{equation*} for every $1\leq k\leq m$, where $C$ is the set of \textquotedblleft coins\textquotedblright\ in $D$.
So, $(\llbracket X_1^{(\infty)}\rrbracket,...,\llbracket X_m^{(\infty)}\rrbracket)$ is a fixed point of $\llbracket D\rrbracket$. On the other hand, if $(\mathbf{A}_1,...,\mathbf{A}_m)\in\mathcal{O}(\mathcal{G}(\mathcal{H}_C)\otimes\mathcal{H})^m$ is a fixed point of $\llbracket D\rrbracket$, then  we can prove that for every $n\geq 0$, $(\llbracket X_1^{(n)}\rrbracket,...,\llbracket X_m^{(n)}\rrbracket)\sqsubseteq (\mathbf{A}_1,...,\mathbf{A}_m)$ by induction on $n$. Indeed, the case of $n=0$ is obvious. In general, using the induction hypothesis on $n-1$, Corollary \ref{corr1} and Lemma \ref{lem-stut} we obtain:
\begin{align*}&(\mathbf{A}_1,...,\mathbf{A}_m)=\llbracket D\rrbracket(\mathbf{A}_1,...,\mathbf{A}_m)\\
&=((\mathbb{K}_C^m\circ\llbracket P_1\rrbracket)(\mathbf{A}_1,...,\mathbf{A}_m),..., (\mathbb{K}_C^m\circ\llbracket P_m\rrbracket)(\mathbf{A}_1,...,\mathbf{A}_m))
\\
&\sqsupseteq ((\mathbb{K}_C^m\circ\llbracket P_1\rrbracket)(\llbracket X^{(n-1)}_1\rrbracket,...,\llbracket X^{(n-1)}_m\rrbracket),..., (\mathbb{K}_C^m\circ\llbracket P_m\rrbracket)(\llbracket X^{(n-1)}_1\rrbracket,...,\llbracket X^{(n-1)}_m\rrbracket))\\ &=(\llbracket X_1^{(n)}\rrbracket,...,\llbracket X_m^{(n)}\rrbracket).
\end{align*} Therefore, it holds that $$(\llbracket X_1^{(\infty)}\rrbracket,...,\llbracket X_m^{(\infty)}\rrbracket)=\bigsqcup_{n=0}^{\infty}(\llbracket X_1^{(n)}\rrbracket,...,\llbracket X_m^{(n)}\rrbracket)\sqsubseteq 
(\mathbf{A}_1,...,\mathbf{A}_m),$$ and $(\llbracket X_1^{(\infty)}\rrbracket,...,\llbracket X_m^{(\infty)}\rrbracket)$ is the least fixed point of $\llbracket D\rrbracket$. 
$\blacksquare$\end{proof}

In light of this theorem, we will simply write $\llbracket P\rrbracket$ for both the denotational (fixed point) and operational semantics of a recursive program $P$. But we should carefully distinguish the semantics   
$\llbracket P\rrbracket\in\mathcal{O}(\mathcal{G}(\mathcal{H}_C)\otimes\mathcal{H})$ of a recursive program $P=P[X_1,...,X_m]$ declared by a system of recursive equations about $X_1,...,X_m$  
from the semantic functional $\llbracket P\rrbracket: \mathcal{O}(\mathcal{G}(\mathcal{H}_C)\otimes\mathcal{H})^m\rightarrow \mathcal{O}(\mathcal{G}(\mathcal{H}_C)\otimes\mathcal{H})$ of program scheme $P=P[X_1,...,X_m]$. Usually, such a difference can be recognised from the context. 

\subsection{Examples}

Now let us reconsider the recursive quantum walks defined in Section \ref{Rewk}.

\begin{exam}\label{ex4} (Unidirectionally recursive Hadamard walk) The semantics of the $n$th approximation of the unidirectionally recursive Hadamard walk specified by equation (\ref{rhw}) is \begin{equation}\label{sem1}\llbracket X^{(n)}\rrbracket =\sum_{i=0}^{n-1}\left[\left(\bigotimes_{j=0}^{i-1} |R\rangle_{d_j}\langle R|\otimes |L\rangle_{d_i}\langle L|\right)\mathbf{H}(i)\otimes T_LT_R^i\right]\end{equation} where $d_0=d$, $\mathbf{H}(i)$ is the operator in $\mathcal{H}_d^{\otimes i}$ defined from the Hadamard operator $H$ by equation (\ref{evo1}). This can be easily shown by induction on $n$, starting from the first three approximations displayed in equation (\ref{sap}). Therefore, the semantics of the unidirectionally recursive Hadamard walk is the operator: \begin{equation}\label{unidw}\begin{split}\llbracket X\rrbracket &=\lim_{n\rightarrow\infty}\llbracket X^{(n)}\rrbracket\\
&=\sum_{i=0}^{\infty}\left[\left(\bigotimes_{j=0}^{i-1} |R\rangle_{d_j}\langle R|\otimes |L\rangle_{d_i}\langle L|\right)\mathbf{H}(i)\otimes T_LT_R^i\right]\\
&=\left[\sum_{i=0}^{\infty}\left(\bigotimes_{j=0}^{i-1} |R\rangle_{d_j}\langle R|\otimes |L\rangle_{d_i}\langle L|\right)\otimes T_LT_R^i\right]\left(\mathbf{H}\otimes I\right)
\end{split}\end{equation} in $\mathcal{F}(\mathcal{H}_d)\otimes\mathcal{H}_p$, where $\mathcal{H}_d=\spa\{L,R\}$, $\mathcal{H}_p=\spa\{|n\rangle:n\in\mathbb{Z}\}$, $I$ is the identity operator in the position Hilbert space $\mathcal{H}_p$, $\mathbf{H}(i)$ is as in equation (\ref{sem1}), and $\mathbf{H}=\sum_{i=0}^\infty\mathbf{H}(i)$ is the extension of $H$ in the free Fock space $\mathcal{F}(\mathcal{H}_d)$ over the direction Hilbert space $\mathcal{H}_d$. \end{exam}

\begin{exam}\label{exam-t} (Bidirectionally recursive Hadamard walk) Let us consider the semantics of the bidirectionally recursive Hadamard walk declared by equation (\ref{ddrhw}). For any string $\Sigma=\sigma_0\sigma_1...\sigma_{n-1}$ of $L$ and $R$, its dual is defined to be $\overline{\Sigma}=\overline{\sigma_0} \overline{\sigma_1}...\overline{\sigma_{n-1}}$, where $\overline{L}=R$ and $\overline{R}=L$. Moreover, we write $|\Sigma\rangle=|\sigma_0\rangle_{d_0}\otimes |\sigma_1\rangle_{d_1}\otimes ...\otimes |\sigma_{n-1}\rangle_{d_{n-1}}$, $T_\Sigma=T_{\sigma_{n-1}}...T_{\sigma_1}T_{\sigma_0}$ and $$\rho_\Sigma=|\Sigma\rangle\langle\Sigma|=\bigotimes_{j=0}^{n-1}|\sigma_j\rangle_{d_j}\langle\sigma_j|.$$ Then the semantics of procedures $X$ and $Y$ are \begin{equation}\label{bidw}\begin{split}
\llbracket X\rrbracket &=\left[\sum_{n=0}^{\infty}\left(\rho_{\Sigma_n}\otimes T_n\right)\right]\left(\mathbf{H}\otimes I_p\right),\\
\llbracket Y\rrbracket &=\left[\sum_{n=0}^{\infty}\left(\rho_{\overline{\Sigma_n}}\otimes T^\prime_n\right)\right]\left(\mathbf{H}\otimes I_p\right),
\end{split}\end{equation} where $\mathbf{H}$ is as in Example \ref{ex4}, and $$\Sigma_n=\begin{cases}(RL)^k L &{\rm if}\ n=2k+1,\\ (RL)^k RR &{\rm if}\ n=2k+2,\end{cases}$$
\begin{equation*}
\begin{split}T_n&=T_{\Sigma_n}=\begin{cases}T_L &{\rm if}\ n\ {\rm is\ odd},\\ T_R^2 &{\rm if}\ n\ {\rm is\ even},\end{cases}\\
T^\prime_n &=T_{\overline{\Sigma_n}}=\begin{cases}T_R &{\rm if}\ n\ {\rm is\ odd},\\ T_L^2 &{\rm if}\ n\ {\rm is\ even}.\end{cases}\end{split}\end{equation*}\end{exam}

It is clear from equations (\ref{unidw}) and (\ref{bidw}) that the behaviours of unidirectionally and bidirectionally recursive Hadamard walks are very different: the unidirectionally one can go to any one of the positions $-1,0,1,2,...$, but the bidirectionally walk $X$ can only go to the positions $-1$ and $2$, and $Y$ can only go to the positions $1$ and $-2$.

\section{Recovering Symmetry and Antisymmetry}\label{Recover}

The solutions of recursive equations found in the free Fock space using the techniques developed in the last section cannot directly apply to the symmetric Fock space for bosons or the antisymmetric Fock space for fermions because they may not preserve symmetry. In this section, we consider symmetrisation of these solutions of recursive equations.  

\subsection{Symmetrisation Functional}

We first examine the domain of symmetric operators in the Fock spaces used in defining semantics of recursive programs. As in Subsection \ref{domain}, let $\mathcal{H}$ be the state Hilbert space of the principal system and $C$ the set of \textquotedblleft coins\textquotedblright, and $$\mathcal{G}(\mathcal{H}_C)\otimes\mathcal{H}=\bigoplus_{\overline{n}\in\omega^C}\left(\bigotimes_{c\in C}\mathcal{F}(\mathcal{H}_c)\otimes\mathcal{H}\right),$$ where $\omega$ is the set of nonnegative integers, and for each $c\in C$, $\mathcal{F}(\mathcal{H}_c)$ is the free Fock space over the state Hilbert space $\mathcal{H}_c$ of \textquotedblleft coin\textquotedblright\ $c$.
For any operator $\mathbf{A}=\sum_{\overline{n}\in\omega^C}\mathbf{A}(\overline{n})\in\mathcal{O}(\mathcal{G}(\mathcal{H}_C)\otimes\mathcal{H})$, we say that $\mathbf{A}$ is symmetric if for each $\overline{n}\in\omega^c$, for each $c\in C$ and for each permutation $\pi$ of $0,1,...,n_c-1$, $P_\pi$ and $\mathbf{A}(\overline{n})$ commute; that is, $$P_\pi\mathbf{A}(\overline{n})=\mathbf{A}(\overline{n})P_\pi.$$ (Note that in the above equation $P_\pi$ actually stands for its cylindrical extension $P_\pi\otimes\bigotimes_{d\in C\setminus\{c\}}I_d\otimes I$ in $\bigotimes_{d\in C}\mathcal{H}_d^{\otimes n_d}\otimes\mathcal{H}$, where $I_d$ is the identity operator in $\mathcal{H}_d$ for every $d\in C\setminus\{c\}$, and $I$ is the identity operator in $\mathcal{H}$.)  
We write $\mathcal{SO}(\mathcal{G}(\mathcal{H}_C)\otimes\mathcal{H})$ for the set of symmetric operators $\mathbf{A}\in \mathcal{O}(\mathcal{G}(\mathcal{H}_C)\otimes\mathcal{H})$.

\begin{lem} $(\mathcal{SO}(\mathcal{G}(\mathcal{H}_C)\otimes\mathcal{H}),\sqsubseteq)$ as a sub-partial order of $(\mathcal{O}(\mathcal{G}(\mathcal{H}_C)\otimes\mathcal{H}),\sqsubseteq)$ is complete.
\end{lem}

\begin{proof} It suffices to observe that symmetry of operators is preserved by the least upper bound in $(\mathcal{O}(\mathcal{G}(\mathcal{H}_C)\otimes\mathcal{H}),\sqsubseteq)$ ; that is, if $\mathbf{A}_i$ is symmetric, so is $\bigsqcup_i\mathbf{A}_i$, as constructed in the proof of Lemma \ref{cpo-d}. $\blacksquare$
\end{proof}

Now we generalise the symmetrisation functional defined by equations (\ref{symz1}) and (\ref{symz2}) into the space $\mathcal{G}(\mathcal{H}_C)\otimes\mathcal{H})$. For each $\overline{n}\in\omega^C$, the symmetrisation functional $\mathbb{S}$ over operators in the space $\bigotimes_{c\in C}\mathcal{H}_c^{\otimes n_c}\otimes\mathcal{H}$ is defined by $$\mathbb{S}(\mathbf{A})=\prod_{c\in C}\frac{1}{n_c!}\cdot \sum_{\{\pi_c\}}\left[\left(\prod_{c\in C}P_{\pi_c}\right)\mathbf{A}\left(\prod_{c\in C}P_{\pi_c}^{-1}\right)\right]$$ for every operator $\mathbf{A}$ in $\bigotimes_{c\in C}\mathcal{H}_c^{\otimes n_c}\otimes\mathcal{H}$, where $\{\pi_c\}$ traverses over all $C-$indexed families with $\pi_c$ being a permutation of $0,1,...,n_c-1$ for every $c\in C$. This symmetrisation functional can be extended to $\mathcal{O}(\mathcal{G}(\mathcal{H}_C)\otimes\mathcal{H})$ in a natural way: 
$$\mathbb{S}(\mathbf{A})=\sum_{\overline{n}\in\omega^C}\mathbb{S}(\mathbf{A}(\overline{n}))$$ for any $\mathbf{A}=\sum_{\overline{n}\in\omega^C}\mathbf{A}(\overline{n})\in\mathcal{O}(\mathcal{G}(\mathcal{H}_C)\otimes\mathcal{H})$. Obviously, $\mathbb{S}(\mathbf{A})\in\mathcal{SO}(\mathcal{G}(\mathcal{H}_C)\otimes\mathcal{H})$. Furthermore, we have:

\begin{lem}\label{Symm-lem} The symmetrisation functional $\mathbb{S}:(\mathcal{O}(\mathcal{G}(\mathcal{H}_C)\otimes\mathcal{H}),\sqsubseteq)\rightarrow (\mathcal{SO}(\mathcal{G}(\mathcal{H}_C)\otimes\mathcal{H}),\sqsubseteq)$ is continuous.\end{lem}

\begin{proof} What we need to prove is that $\mathbb{S}\left(\bigsqcup_{i}\mathbf{A}_i\right)=\bigsqcup_i\mathbb{S}(\mathbf{A}_i)$ for any chain $\{\mathbf{A}_i\}$ in $(\mathcal{O}(\mathcal{G}(\mathcal{H}_C)\otimes\mathcal{H}),\sqsubseteq)$. Assume that $\mathbf{A}=\bigsqcup_i\mathbf{A}_i$. Then by the proof of Lemma \ref{cpo-d}, we can write $\mathbf{A}=\sum_{\overline{n}\in\omega}\mathbf{A}(\overline{n})$ and $\mathbf{A}_i=\sum_{\overline{n}\in\Omega_i}\mathbf{A}(\overline{n})$ for some $\Omega_i$ with $\sup_i\Omega_i=\omega^C$. So, it holds that 
$$\bigsqcup_i\mathbb{S}(\mathbf{A}_i)=\bigsqcup_i\sum_{\overline{n}\in\Omega_i}\mathbb{S}(\mathbf{A}(\overline{n}))=\sum_{\overline{n}\in\omega^C}\mathbb{S}(\mathbf{A}(\overline{n}))=\mathbb{S}(\mathbf{A}).\ \blacksquare$$\end{proof}

\subsection{Symmetrisation of the Semantics of Recursive Programs}

Now we are ready to present the semantics of recursive programs in the symmetric or antisymmetric Fock space. 

\begin{defn} Let $P=P[X_1,...,X_m]$ be a recursive program declared by the system $D$ of recursive equations (\ref{Decla}). Then its symmetric semantics $\llbracket P\rrbracket_{sym}$ is the symmetrisation of its semantics $\llbracket P\rrbracket$ in the free Fock space:
$$\llbracket P\rrbracket_{sym}=\mathbb{S}(\llbracket P\rrbracket)$$ where $\llbracket P\rrbracket=\llbracket P\rrbracket_{fix}=\llbracket P\rrbracket_{op}\in\mathcal{O}(\mathcal{G}(\mathcal{H}_C)\otimes\mathcal{K})$ (see Theorem \ref{ss-equv}), $C$ is the set of \textquotedblleft coins\textquotedblright\ in $D$, and $\mathcal{H}$ is the state Hilbert space of the principal system of $D$.  
\end{defn}

As a symmetrisation of Proposition \ref{prop-sapp}, we have:\begin{prop}$\llbracket P\rrbracket_{sym}=\bigsqcup_{n=0}^\infty\mathbb{S}(\llbracket P^{(n)}\rrbracket)$.\end{prop}

\begin{proof} It follows from Proposition \ref{prop-sapp} and Lemma \ref{Symm-lem} (continuity of the symmetrisation functional) that \begin{equation*}\llbracket P\rrbracket_{sym}
=\mathbb{S}(\llbracket P\rrbracket)=\mathbb{S}\left(\bigsqcup_{n=0}^\infty\llbracket P^{(n)}\rrbracket\right)
=\bigsqcup_{n=0}^\infty\mathbb{S}(\llbracket P^{(n)}\rrbracket).\ \blacksquare\end{equation*}\end{proof}

Again, let us consider the examples of recursive Hadamard walks. 

\begin{exam} (Continuation of Example \ref{ex4}) For each $i\geq 0$, we have: \begin{align*}\mathbb{S}&\left(\bigotimes_{j=0}^{i-1}|R\rangle_{d_j}\langle R|\otimes |L\rangle_{d_i}\langle L|\right)
=\frac{1}{(i+1)!}\sum_\pi P_\pi\left(\bigotimes_{j=0}^{i-1}|R\rangle_{d_j}\langle R|\otimes |L\rangle_{d_i}\langle L|\right)P_\pi^{-1}\\ &\ \ \ \ \ \ \ \ \ \ \ \ \ \ \ \ \ \ \ \ \ \ \ \ \ \ \ \ \ \ \ \ \ \ \ \ \ \ \ \ \ \ \ \ \ \ ({\rm where}\ \pi\ {\rm traverses\ over\ all\ permutations\ of}\ 0,1,...,i)\\ &=\frac{1}{i+1}\sum_{j=0}^i\left(|R\rangle_{d_0}\langle R|\otimes...\otimes|R\rangle_{d_{j-1}}\langle R|\otimes |L\rangle_{d_j}\langle L|\otimes |R\rangle_{d_{j+1}}\langle R|\otimes...\otimes
|R\rangle_{d_i}\langle R|\right)\\ &\stackrel{\triangle}{=}G_i. 
\end{align*}
Therefore, the symmetric semantics of the unidirectionally recursive Hadamard walk defined by equation (\ref{rhw}) is $$\mathbb{S}(\llbracket X\rrbracket)=\left(\sum_{i=0}^\infty G_i\otimes T_LT_R^i\right)(\mathbf{H}\otimes I).$$\end{exam}

\begin{exam} (Continuation of Example \ref{exam-t}) The symmetric semantics of the bidirectionally recursive Hadamard walk specified by equaltion (\ref{ddrhw}) is: 
\begin{align*}\llbracket X\rrbracket &= \left[\sum_{n=0}^\infty(\gamma_n\otimes T_n)\right]\left(\mathbf{H}\otimes I_p\right),\\
\llbracket Y\rrbracket &= \left[\sum_{n=0}^\infty(\delta_n\otimes T_n)\right]\left(\mathbf{H}\otimes I_p\right)
\end{align*} where: $$\gamma_{2k+1}=\frac{1}{\left(\begin{array}{cc}k\\ 2k+1\end{array}\right)}\sum_\Gamma\rho_\Gamma,\ \ \ \ \ \ \delta_{2k+1}=\frac{1}{\left(\begin{array}{cc}k\\ 2k+1\end{array}\right)}\sum_\Delta\rho_\Delta$$ with $\Gamma$ ranging over all strings of $(k+1)$ $L$'s and $k$ $R$'s and $\Delta$ ranging over all strings of $k$ $L$'s and $(k+1)$ $R$'s, and $$\gamma_{2k+2}=\frac{1}{\left(\begin{array}{cc}k\\ 2k+2\end{array}\right)}\sum_\Gamma\rho_\Gamma,\ \ \ \ \ \ \sigma_{2k+2}=\frac{1}{\left(\begin{array}{cc}k\\ 2k+2\end{array}\right)}\sum_\Delta\rho_\Delta$$ with $\Gamma$ ranging over all strings of $k$ $L$'s and $(k+2)$ $R$'s and $\Delta$ ranging over all strings of $(k+2)$ $L$'s and $k$ $R$'s.
\end{exam}

\subsection{Principal System Semantics of Quantum Recursion}

Let $P$ be a recursive program with $\mathcal{H}$ being the state Hilbert space of its principal variables and $C$ being the set of its \textquotedblleft coins\textquotedblright. We consder the computation of $P$ with input $|\psi\rangle\in\mathcal{H}$. Assume that the \textquotedblleft coins\textquotedblright\ are initialised in state $|\Psi\rangle\in\bigotimes_{c\in C}\mathcal{F}_{v_c}(\mathcal{H}_c)$, where for each $c\in C$, $\mathcal{H}_c$ is the state Hilbert space of \textquotedblleft coin\textquotedblright\ $c$, $\mathcal{F}_{v_c}(\mathcal{H}_c)$ is the symmetric or antisymmetric Fock space over $\mathcal{H}_c$, and $v_c=+$ or $-$ if \textquotedblleft coin\textquotedblright\ $c$ is implemented by a boson or a fermion, respectively. Then the computation of the program starts in state $|\Psi\rangle|\psi\rangle$. What actually concerns us is the output in the principal system. This observation leads to the following:

\begin{defn}\label{principal} Given a state $|\Psi\rangle\in\bigotimes_{c\in C}\mathcal{F}_{v_c}(\mathcal{H}_c)$. The principal system semantics of program $P$ with respect to \textquotedblleft coin\textquotedblright\ initialisation $|\Psi\rangle$ is the mapping $\llbracket P,\Psi\rrbracket$ from pure states in $\mathcal{H}$ to partial density operators \cite{Se04}, i.e. positive operators with trace $\leq 1$, in $\mathcal{H}$:  $$\llbracket P,\Psi\rrbracket (|\psi\rangle)=tr_{\bigotimes_{c\in C}\mathcal{F}_{v_c}(\mathcal{H}_c)}(|\Phi\rangle\langle\Phi|)$$ for each pure state $|\psi\rangle$ in $\mathcal{H}$, where $$|\Phi\rangle=\llbracket P\rrbracket_{sym} (|\Psi\rangle\otimes |\psi\rangle),$$ $\llbracket P\rrbracket_{sym}$ is the symmetric semantics of $P$, and $tr_{\bigotimes_{c\in C}\mathcal{F}_{v_c}(\mathcal{H}_c)}$ is the partial trace over $\bigotimes_{c\in C}\mathcal{F}_{v_c}(\mathcal{H}_c)$ (see \cite{NC00}, Section 2.4.3).
\end{defn}

\begin{exam} (Continuation of Example \ref{exam-t}) We consider the bidirectionally recursive Hadamard walk declared by equation (\ref{ddrhw}) once again and suppose that it starts from the position $0$.   

\begin{enumerate}\item If the \textquotedblleft coins\textquotedblright\ are bosons initialised in state $$|\Psi\rangle=|L,L,...,L\rangle_+ =|L\rangle_{d_0}\otimes |L\rangle_{d_1}\otimes ... \otimes |L\rangle_{d_{n-1}},$$ then we have \begin{equation*}\begin{split}
&\llbracket X\rrbracket_{sym} (|\Psi\rangle\otimes |0\rangle)= \begin{cases}\frac{1}{\sqrt{2^{n}}\left(\begin{array}{cc}k\\ 2k+1\end{array}\right)}\sum_{\Gamma}|\Gamma\rangle\otimes |-1\rangle &{\rm if}\ n=2k+1,\\ \frac{1}{\sqrt{2^{n}}\left(\begin{array}{cc}k\\ 2k+2\end{array}\right)}\sum_\Delta |\Delta\rangle\otimes |2\rangle &{\rm if}\ n=2k+2,
\end{cases}\end{split}\end{equation*} where $\Gamma$ traverses over all strings of $(k+1)$ $L$'s and $k$ $R$'s, and $\Delta$ traverses over all strings of $k$ $L$'s and $(k+2)$ $R$'s.
Therefore, the principal system semantics with the \textquotedblleft coin\textquotedblright\ initialisation $|\Psi\rangle$ is: $$\llbracket X,\Psi\rrbracket (|0\rangle)=\begin{cases}
\frac{1}{2^n} |-1\rangle\langle -1| &{\rm if}\ n\ {\rm is\ odd},\\ \frac{1}{2^n}|2\rangle\langle 2| &{\rm if}\ n\ {\rm is\ even}.
\end{cases}$$

\item Recall from \cite{MR04} that for each single-particle state $|\psi\rangle$ in $\mathcal{H}_d$, the corresponding coherent state of bosons in the symmetric Fock space $\mathcal{F}_+(\mathcal{H}_d)$ over $\mathcal{H}_d$ is defined as $$|\psi\rangle_{coh}=\exp{\left(-\frac{1}{2}\langle\psi |\psi\rangle\right)}\sum_{n=0}^\infty\frac{[a^\dag(\psi)]^n}{n!}|\mathbf{0}\rangle$$ where $|\mathbf{0}\rangle$ is the vacuum state and $a^\dag(\cdot)$ the creation operator. 
If the \textquotedblleft coins\textquotedblright\ are initialised in the coherent state $|L\rangle_{coh}$ of bosons corresponding to $|L\rangle$, then we have:
 \begin{equation*}\begin{split}\llbracket X\rrbracket _{sym}(&|L\rangle_{{\rm coh}}\otimes |0\rangle)= \frac{1}{\sqrt{e}} \left(\sum_{k=0}^\infty \frac{1}{\sqrt{2^{2k+1}}\left(\begin{array}{cc}k\\ 2k+1\end{array}\right)}\sum_{\Gamma_k}|\Gamma_k\rangle\right)\otimes |-1\rangle\\ &\ \ \ \ \ \ \ \ \ \ \ +\frac{1}{\sqrt{e}}\sum_{k=0}^\infty\left(\frac{1}{\sqrt{2^{2k+2}}
 \left(\begin{array}{cc}k\\ 2k+2\end{array}\right)}\sum_{\Delta_k}|\Delta_k\rangle\right)\otimes |2\rangle,\end{split}\end{equation*}
 where $\Gamma_k$ ranges over all strings of $(k+1)$ $L$'s and $k$ $R$'s, and $\Delta_k$ ranges over all strings of $k$ $L$'s and $(k+2)$ $R$'s.
 So, the principal system semantics with \textquotedblleft coin\textquotedblright\ initialisation $|L\rangle_{{\rm coh}}$ is:
\begin{equation*}\begin{split}\llbracket X, L_{{\rm coh}}\rrbracket(|0\rangle)&=\frac{1}{\sqrt{e}}\left(\sum_{k=0}^\infty
\frac{1}{2^{2k+1}} |-1\rangle\langle -1| +\sum_{k=0}^\infty \frac{1}{2^{2k+2}}|2\rangle\langle 2|\right)\\ &=\frac{1}{\sqrt{e}}\left(\frac{2}{3} |-1\rangle\langle -1| +\frac{1}{3}|2\rangle\langle 2|\right).
\end{split}\end{equation*}\end{enumerate}
\end{exam}

\section{Quantum Loop}\label{QWLO}

In this section, we consider a special class of quantum recursions. Arguably, while-loop is the simplest and most popular form of recursion used in programming languages. In classical programming, the while-loop $$\mathbf{while}\ b\ \mathbf{do}\ S\ \mathbf{od}$$ can be seen as the program $X$ declared by the recursive equation:\begin{equation}\label{while1} X\Leftarrow \mathbf{if}\ b\ \mathbf{then}\ X\ \mathbf{else}\ \mathbf{skip}\ \mathbf{fi}\end{equation} We can define a kind of quantum while-loop by using quantum case statement and quantum choice in the place of classical case statement $\mathbf{if} ... \mathbf{then} ...\mathbf{else\ fi}$ in equation (\ref{while1}).

\begin{exam}\label{awhile} (Quantum while-loop)\begin{enumerate}\item The first form of quantum while-loop: \begin{equation}\label{while2}\mathbf{qwhile}\ [c]=|1\rangle\ \mathbf{do}\ U[q]\ \mathbf{od}\end{equation}
is defined to be the recursive program $X$ declared by \begin{equation}\label{recq1}\begin{split}X\Leftarrow\ &\mathbf{qif}[c]\ |0\rangle \rightarrow \mathbf{skip}\\ &\ \ \ \ \ \ \ \square\ |1\rangle \rightarrow U[q];X\\ &\mathbf{fiq}
\end{split}\end{equation} where $c$ is a quantum \textquotedblleft coin\textquotedblright\ variable denoting a qubit, $q$ is a principal quantum variable, and $U$ is a unitary operator in the state Hilbert space $\mathcal{H}_q$ of system $q$.
\item The second form of quantum while-loop \begin{equation}\label{while3}\mathbf{qwhile}\ V[c]=|1\rangle\ \mathbf{do}\ U[q]\ \mathbf{od}\end{equation}
is defined to be the recursive program $X$ declared by
\begin{equation}\label{recq2}\begin{split}X\Leftarrow\  & \mathbf{skip}\oplus_{V[c]} (U[q];X)\\ &\equiv V[c]; \mathbf{qif}[c]\ |0\rangle \rightarrow \mathbf{skip}\\ &\ \ \ \ \ \ \ \ \ \ \ \ \ \ \ \ \ \ \ \ \ \square\ |1\rangle \rightarrow U[q];X\\ &\ \ \ \ \ \ \ \ \ \ \ \ \ \ \ \mathbf{fiq}
\end{split}\end{equation} Note that the recursive equation (\ref{recq2}) is obtained by replacing the quantum case statement $\mathbf{qif} ...\mathbf{fiq}$ in equation (\ref{recq1}) by the quantum choice $\oplus_{V[c]}$.
\item Actually, quantum loops (\ref{while2}) and (\ref{while3}) are not very interesting because there is not any interaction between the quantum \textquotedblleft coin\textquotedblright\ and the principal quantum system $q$ in them. This situation is corresponding to the trivial case of classical loop (\ref{while1}) where the loop guard $b$ is irrelevant to the loop body $S$. The classical loop (\ref{while1}) becomes truly interesting only when the loop guard $b$ and the loop body $S$ share some program variables. Likewise, a much more interesting form of quantum while-loop is
\begin{equation}\label{while4}\mathbf{qwhile}\ W[c;q]=|1\rangle\ \mathbf{do}\ U[q]\ \mathbf{od}\end{equation} which is defined to be the program $X$ declared by the recursive equation
\begin{equation*}\begin{split}X\Leftarrow\ W[c,q];\ &\mathbf{qif}[c]\ |0\rangle\rightarrow\mathbf{skip}\\ &\ \ \ \ \ \ \ \square\ |1\rangle\rightarrow U[q];X\\ &\mathbf{fiq}\end{split}\end{equation*}
where $W$ is a unitary operator in the state Hilbert space $\mathcal{H}_c\otimes\mathcal{H}_q$ of the composed system of the quantum \textquotedblleft coin\textquotedblright\ $c$ and the principal system $q$. The operator $W$ describes the interaction between the \textquotedblleft coin\textquotedblright\ $c$ and the principal system $q$.
It is obvious that the loop (\ref{while4}) degenerates to the loop (\ref{while3}) whenever $W=V\otimes I$, where $I$ is the identity operator in $\mathcal{H}_q$.
The semantics of the loop (\ref{while4}) in the free Fock space is the operator: \begin{equation*}\begin{split}
\llbracket X\rrbracket &=\sum_{k=1}^\infty (|1\rangle_{c_0}\langle 1|\otimes (|1\rangle_{c_1}\langle 1|\otimes ... (|1\rangle_{c_{k-2}}\langle 1|\otimes (|0\rangle_{c_{k-1}}\langle 0|
\otimes U^{k-1}[q])\\ &\ \ \ \ \ \ \ \ \ \ \ \ \ \ \ \ \ \ \ \ \ \ \ \ \ \ \ \ \ \ \ \ \ \ \ \ W[c_{k-1},q])W[c_{k-2},q]...)W[c_1,q])W[c_0,q]\\
&=\sum_{k=1}^\infty\left[\left(\bigotimes_{j=0}^{k-2} |1\rangle_{c_j}\langle 1|\otimes |0\rangle_{c_{k-1}}\langle 0|\otimes U^{k-1}[q]\right)\prod_{j=0}^{k-1}W[c_j,q]\right].
\end{split}
\end{equation*} Furthermore, the symmetric semantics of the loop is: 
\begin{equation*}
\llbracket X\rrbracket_{sym}=\sum_{k=1}^\infty\left[\left(\mathbf{A}(k)\otimes U^{k-1}[q]\right)\prod_{j=0}^{k-1}W[c_j,q]\right],
\end{equation*} where: $$\mathbf{A}(k)=\frac{1}{k}\sum_{j=0}^{k-1} |1\rangle_{c_0}\langle 1|\otimes...\otimes |1\rangle_{c_{j-1}}\langle 1|\otimes |0\rangle_{c_{j}}\langle 0|\otimes |1\rangle_{c_{j+1}}\langle 1|\otimes ...\otimes |1\rangle_{c_{k-1}}\langle 1|.$$
\end{enumerate}\end{exam}

\section{Conclusion}\label{CON1}

In this paper, we introduced the notion of quantum recursion based on quantum case statement and quantum choice defined in \cite{YYF12}, \cite{YYF13}. Recursive quantum walks and quantum while-loops were presented as examples of quantum recursion. The denotational and operational semantics of quantum recursion were defined by using second quantisation, and they were proved to be equivalent. But we are still at the very beginning of the studies of quantum recursion, and a series of problems are left unsolved:
\begin{itemize}\item First of all, it is not well understood what kind of computational problems can be solved more conveniently by using quantum recursion. 

\item Second, how to build a Floyd-Hoare logic for quantum while-loops defined in Example \ref{awhile}? Blute, Panangaden and Seely \cite{BPS94} observed that Fock space can serve as a model of linear logic with exponential types. Perhaps, such a program logic can be established through combining linear logic with the techniques developed in \cite{Yin11}. 

\item Another important open question is: what kind of physical systems can be used to implement quantum recursion where new \textquotedblleft coins\textquotedblright\ must be continuously created?

\item Finally, we even do not fully understand how does a quantum recursion use its \textquotedblleft coins\textquotedblright\ in its computational process. In the definition of the principal system semantics of a recursive program (Definition \ref{principal}), a state $|\Psi\rangle$ in the Fock space of \textquotedblleft coins\textquotedblright\ is given \textit{a priori}. This means that the states of a \textquotedblleft coin\textquotedblright\ and its copies are given once for all. Another possibility is that the states of the copies of a \textquotedblleft coin\textquotedblright\ are created step by step, as shown in the following:   
\begin{exam}Consider the recursive program $X$ declared by \begin{align*}X\Leftarrow a^\dag_c(|0\rangle);R_y[c,p];\ &\mathbf{qif}\ [c]\ |0\rangle\rightarrow \mathbf{skip}\\
&\ \ \ \ \ \ \ \square\ |1\rangle\rightarrow T_R[p];X\\ &\mathbf{fiq}
\end{align*}
where $c$ is a \textquotedblleft coin\textquotedblright\ variable with state space $\mathcal{H}_c=\spa\{|0\rangle,|1\rangle\}$, the variable $p$ and operator $T_R$ are as in the Hadamard walk, $$R_y[c,p]=\sum_{n=0}^\infty\left[R_y\left(\frac{\pi}{2^{n+1}}\right)\otimes |n\rangle_p\langle n|\right]$$
and $R_y(\theta)$ is the rotation of a qubit about the $y-$axis in the Bloch sphere. Intuitively, $R_y[c,p]$ is a controlled rotation where position of $p$ is used to determine the rotated angle. It is worth noting that this program $X$ is a quantum loop defined in equation (\ref{while4}) but modified by adding a creation operator at the beginning. Its initial behaviour starting at position $0$ with the \textquotedblleft coin\textquotedblright\ $c$ being in the vacuum state $|\mathbf{0}\rangle$ is visualised by the following transitions: 
\begin{align*}
|\mathbf{0}\rangle|0\rangle_p&\stackrel{a^\dag_d(|0\rangle)}{\longrightarrow} |0\rangle|0\rangle_p\stackrel{R_x[d,p]}{\longrightarrow}\frac{1}{\sqrt{2}}\left(|0\rangle+|1\rangle\right)|0\rangle_p\\
&\stackrel{\mathbf{qif}...\mathbf{fiq}}{\longrightarrow}\frac{1}{\sqrt{2}}\left[\left(E,|0\rangle|0\rangle_p\right)+\left(X,|1\rangle|1\rangle_p\right)\right].
\end{align*} The first configuration at the end of the above equation terminates, but the second continues the computation as follows: $$|1\rangle|1\rangle_p\stackrel{a^\dag_d(|0\rangle)}{\longrightarrow} |0,1\rangle_v|0\rangle_p\stackrel{R_x[d,p]}{\longrightarrow}\cdots.$$
\end{exam} It is clear from the above example that the computation of a recursive program with the creation operator is very different from that without it. 
A careful study of quantum recursions that allows the creation operator appear in their syntax will be carried out in another paper.  
\end{itemize}

\section*{Acknowledgement} I'm very grateful to Professor Prakash Panangaden for teaching me the second quantisation method during his visit at the University of Technology, Sydney in 2013. 
The first version of this paper is the text of the third part of my talk \textquotedblleft Quantum programming: from superposition of data to superposition of programs\textquotedblright\ at the Tsinghua Software Day, April 21-22, 2014 (see: http://sts.thss.tsinghua.edu.cn/tsd2014/home.html. The first part of the talk is based on \cite{Yin11}, and the second part is based on \cite{YYF13}).
I'm also grateful to Professors Jean-Pierre Jouannaud and Ming Gu for inviting me.

\end{document}